\documentclass[runningheads]{llncs}

\usepackage{stmaryrd, amssymb, extarrows}
\usepackage{verbatim} 

\usepackage{hyperref}

\newcommand{\nc}{\newcommand}

\nc{\toy}{\mathcal{TOY}} 

\nc{\NAT}{\mathbb{N}}
\nc{\REAL}{\mathbb{R}}

\nc{\tup}[1]{\overline{#1}}   
\nc{\ntup}[2]{\tup{#1}_{#2}}  

\newcommand{\qdom}{\mathcal{D}} 
\newcommand{\aqdom}{D_\qdom \setminus \{\bt\}} 
\nc{\bt}{\mathrm{\mathbf{b}}} 
\nc{\tp}{\mathrm{\mathbf{t}}} 
\nc{\dleq}{\trianglelefteqslant} 
\nc{\dlt}{\vartriangleleft} 
\nc{\dgeq}{\trianglerighteqslant} 
\nc{\dgt}{\vartriangleright} 

\newcommand{\B}{\mathcal{B}} 
\newcommand{\U}{\mathcal{U}} 
\newcommand{\W}{\mathcal{W}} 

\nc{\leqinfo}{\sqsubseteq}
\nc{\lsinfo}{\sqsubset}
\nc{\geqinfo}{\sqsupseteq}
\nc{\gtinfo}{\sqsupset}

\newcommand{\cdom}{\mathcal{C}} 
\newcommand{\rdom}{\mathcal{R}} 

\newcommand{\transform}[1]{{#1}^\mathcal{T}} 

\nc{\Unsat}[2]{\mbox{Unsat}_{#1}(#2)}

\nc{\Exp}{\mbox{Exp}_{\bot}(\Sigma,B,\Var)} 
\nc{\TExp}{\mbox{Exp}(\Sigma,B,\Var)} 
\nc{\GExp}{\mbox{Exp}_{\bot}(\Sigma,B)} 
\nc{\TGExp}{\mbox{Exp}(\Sigma,B)} 
\nc{\Term}{\mbox{Term}_{\bot}(\Sigma,B,\Var)} 
\nc{\TTerm}{\mbox{Term}(\Sigma,B,\Var)} 
\nc{\GTerm}{\mbox{Term}_{\bot}(\Sigma,B)} 
\nc{\TGTerm}{\mbox{Term}(\Sigma,B)} 

\nc{\Sol}[2]{\mbox{Sol}_{#1}(#2)} 
\nc{\Solc}[1]{\Sol{\cdom}{#1}} 

\nc{\Int}[2]{\mbox{Int}_{#1,#2}} 
\nc{\Intdc}{\Int{\qdom}{\cdom}} 
\nc{\ibot}{\bot\!\!\!\bot} 
\nc{\itop}{\top\!\!\!\top} 

\nc{\pc}[2]{\mathsf{#1}(#2)} 
\nc{\qval}[1]{\pc{qVal}{#1}} 
\nc{\qbound}[1]{\pc{qBound}{#1}} 

\nc{\encode}[1]{\ulcorner#1\urcorner} 

\nc{\scheme}[2]{\text{#1}(#2)} 
\nc{\qlp}[1]{\scheme{QLP}{#1}} 
\nc{\bqlp}[1]{BQLP({#1})} 
\nc{\clp}[1]{\scheme{CLP}{#1}} 
\nc{\cflp}[1]{\scheme{CFLP}{#1}} 
\nc{\qcflp}[2]{\scheme{QCFLP}{#1,#2}} 

\nc{\diff}{~{\Longleftrightarrow_{\mathrm{def}}}~} 
\nc{\eqdef}{~{=_{\mathrm{def}}}~} 
\nc{\union}{\bigcup} 
\nc{\inter}{\bigcap} 
\nc{\supr}{\bigsqcup} 
\nc{\infi}{\bigsqcap} 
\nc{\entail}[1]{~{\succcurlyeq_{#1}}~} 
\nc{\model}[1]{~{\models_{#1}}~} 

\nc{\Prog}{\mathcal{P}} 
\nc{\Var}{\mathcal V\!ar} 
\nc{\War}{\mathcal W\!ar} 
\nc{\set}[2]{\mbox{#1}(#2)} 
\nc{\domset}[1]{\set{vdom}{#1}} 
\nc{\ranset}[1]{\set{vran}{#1}} 
\nc{\varset}[1]{\set{var}{#1}} 
\nc{\bvarset}[1]{\set{bvar}{#1}} 
\nc{\warset}[1]{\set{war}{#1}} 
\nc{\qvalset}[1]{\mathrm{qVal}(#1)} 
\nc{\M}[1]{\mathcal{M}_{#1}} 
\nc{\Mp}{\M{\Prog}} 
\nc{\Sp}{S_\Prog} 
\nc{\STp}{ST_\Prog} 
\newcommand{\I}{\mathcal{I}} 
\nc{\qto}[1]{\xrightarrow{#1}} 
\nc{\qgets}[1]{\xleftarrow{#1}} 
\nc{\sep}{~{\talloblong}~} 
\nc{\closure}[2]{\mbox{cl}_{#1}(#2)} 
\nc{\true}{\texttt{tt}} 
\nc{\false}{\texttt{ff}} 

\nc{\cneq}{~{\!/\!\!=}~} 
\nc{\X}{\mathcal{X}} 

\nc{\infx}{\vdash} 
\nc{\infxx}{\vdash\!\!\vdash} 

\nc{\QHL}{\text{QHL}} 
\nc{\qhl}[1]{\infx_{\QHL(#1)}} 
\nc{\qhln}[2]{\infx_{\QHL(#1)}^{#2}} 

\nc{\CRWL}[1]{\text{CRWL}(#1)} 
\nc{\crwl}[1]{\infx_{#1}} 
\nc{\crwln}[2]{\infx_{#1}^{#2}} 
\nc{\CRWLc}{\CRWL{\cdom}} 
\nc{\crwlc}{\crwl{\cdom}} 

\nc{\ICRWL}[1]{\I\text{-}\CRWL{#1}} 
\nc{\icrwl}[1]{\infxx_{#1}} 
\nc{\icrwln}[2]{\infxx_{#1}^{#2}} 
\nc{\ICRWLc}{\ICRWL{\cdom}} 
\nc{\icrwlc}{\icrwl{\cdom}} 

\nc{\QCRWL}[2]{\text{QCRWL}(#1,#2)} 
\nc{\qcrwl}[2]{\infx_{#1,#2}} 
\nc{\qcrwln}[3]{\infx_{#1,#2}^{#3}} 
\nc{\QCRWLdc}{\QCRWL{\qdom}{\cdom}} 
\nc{\qcrwldc}{\qcrwl{\qdom}{\cdom}} 

\nc{\IQCRWL}[2]{\I\text{-}\QCRWL{#1}{#2}} 
\nc{\iqcrwl}[2]{\infxx_{#1,#2}} 
\nc{\iqcrwln}[3]{\infxx_{#1,#2}^{#3}} 
\nc{\IQCRWLdc}{\IQCRWL{\qdom}{\cdom}} 
\nc{\iqcrwldc}{\iqcrwl{\qdom}{\cdom}} 

\nc{\resx}{\Vdash} 
\nc{\qres}[1]{\resx_{#1}} 
\nc{\qresn}[2]{\resx_{#1}^{#2}} 

\nc{\fracc}[2]{\begin{array}{c}{#1}\\ \hline {#2} \end{array}}  

\title{A Generic Scheme for Qualified \\ Constraint Functional Logic Programming\thanks{Research partially supported by projects MERIT--FORMS (TIN2005-09027-C03-03), PROMESAS--CAM(S-0505/TIC/0407) and STAMP (TIN2008-06622-C03-01).}}
\titlerunning{A Generic Scheme for QCFLP}
\subtitle{Technical Report SIC-1-09}
\author{Rafael Caballero, Mario Rodr\'iguez-Artalejo and Carlos A. Romero-D\'iaz}
\authorrunning{R.~Caballero, M.~Rodr\'iguez-Artalejo and C.A.~Romero-D\'iaz}
\institute{Departamento de Sistemas Inform\'aticos y Computaci\'on, Universidad Complutense,\\
    Facultad de Inform\'atica, 28040 Madrid, Spain\\
    \email{\{rafa,mario\}@sip.ucm.es} and \email{cromdia@fdi.ucm.es}}

\begin{document}
\maketitle

\begin{abstract}
Qualification has been recently introduced as a generalization of uncertainty in the field of Logic Programming. In this report we investigate a more expressive language for First-Order Functional Logic Programming with Constraints and Qualification. We present a Rewriting Logic which characterizes the intended semantics of programs, and a prototype implementation based on a semantically correct program transformation. Potential applications of the resulting language include flexible information retrieval. As a concrete  illustration, we show how to write program rules to compute qualified answers for user queries concerning the books available in a given library. \\

\end{abstract}

\section{Introduction}\label{Introduction}

Various extensions of Logic Programming  with uncertain reasoning capabilities have been widely  investigated during the last 25 years. The recent recollection \cite{Sub07} reviews the evolution of the subject from the viewpoint of a committed researcher. All the proposals in the field replace classical two-valued logic by some kind of many-valued logic with more than two truth values, which are attached to computed answers and interpreted as truth degrees.

In a recent work \cite{RR08,RR08TR} we have presented a {\em Qualified Logic Programming} scheme $\qlp{\qdom}$ parameterized by a {\em qualification domain} $\qdom$, a lattice of so-called {\em qualification values} that are attached to computed answers and  interpreted as a measure of the satisfaction of certain user expectations. $\qlp{\qdom}$-programs are sets of clauses of the form $A \qgets{\alpha} \tup{B}$, where the head $A$ is an atom, the body $\tup{B}$ is a conjunction of atoms, and $\alpha \in \qdom$ is called {\em attenuation} factor. Intuitively, $\alpha$ measures the maximum confidence placed on  an inference performed by the clause. More precisely, any successful application of the clause attaches to the head a qualification value which cannot exceed the infimum of $\alpha \circ \beta_i \in \qdom$, where $\beta_i$ are the qualification values computed for the body atoms and $\circ$ is a so-called  {\em attenuation operator}, provided by $\qdom$.

Uncertain Logic Programming can be expressed by particular instances of $\qlp{\qdom}$, where the user expectation is understood as a lower bound for the {\em truth degree} of the computed answer and $\qdom$ is chosen to formalize a lattice of non-classical truth values. Other choices of $\qdom$  can be designed to model other kinds of user expectations, as e.g. an upper bound for the {\em size} of the logical proof underlying the computed answer. As shown in \cite{CRR08}, the $\qlp{\qdom}$ scheme is also well suited to deal with Uncertain Logic Programming  based on similarity relations in the line of  \cite{Ses02}. Therefore, Qualified Logic Programming has a potential for flexible information retrieval applications, where the answers computed for user queries may match the user expectations only to some degree. As shown in  \cite{RR08}, several useful  instances of $\qlp{\qdom}$ can be conveniently implemented by using constraint solving techniques.

In this report we investigate  an extension of  $\qlp{\qdom}$  to a more expressive scheme, supporting computation with first-order lazy functions and constraints. More precisely, we consider the first-order fragment of $\cflp{\cdom}$, a generic scheme for functional logic programming with constraints over a parametrically given domain $\cdom$ presented in \cite{LRV07}. We propose an extended scheme $\qcflp{\qdom}{\cdom}$ where the additional parameter $\qdom$ stands for a qualification domain. $\qcflp{\qdom}{\cdom}$-programs are sets of conditional rewrite rules of the form $f(\ntup{t}{n}) \qto{\alpha} r \Leftarrow \Delta$, where the condition $\Delta$ is a conjunction of $\cdom$-constraints that may involve user defined functions, and $\alpha \in \qdom$ is an attenuation factor. As in the logic programming case, $\alpha$ measures the maximum confidence placed on an inference performed by the rule: any successful application of the rule attaches to the computed result a qualification value which cannot exceed the infimum of $\alpha \circ \beta_i \in \qdom$, where $\beta_i$ are the qualification values computed for $r$ and $\Delta$, and $\circ$ is $\qdom$'s attenuation operator. $\qlp{\qdom}$ program clauses can be easily formulated as a particular case of $\qcflp{\qdom}{\cdom}$ program rules.

As far as we know, no  related work covers the expressivity of our approach.
Guadarrama et al. \cite{GMV04} have proposed to use real arithmetic constraints as an implementation tool for a Fuzzy Prolog, but their language does not support constraint programming as such.
Starting from the field of natural language processing, Riezler \cite{Rie96,Rie98phd}
has developed quantitative and probabilistic extensions of the classical $\clp{\cdom}$ scheme with the aim of computing good parse trees for  constraint logic grammars, but his work bears no relation to functional programming.
Moreno and Pascual \cite{MP07} have investigated similarity-based unification in the context of {\em needed narrowing} \cite{AEH00}, a narrowing strategy using so-called {\em definitional trees} that underlies the operational semantics of functional logic languages such as {\sf Curry} \cite{curry} and $\toy$ \cite{toy}, but they use neither constraints nor attenuation factors
and they provide no declarative semantics.
The approach of the present report is quite different. We work with a class of programs more general and expressive than the {\em inductively sequential} term rewrite systems used in \cite{MP07}, and our results focus on a rewriting logic used to characterize declarative semantics and to prove the correctness of an implementation technique based on a program transformation. Similarity relations could be easily incorporated to our scheme by using the techniques presented in \cite{CRR08} for the Logic Programming case. Moreover, the good properties of needed narrowing as a computation model are not spoiled by our implementation technique, because our program transformation preserves the structure of the definitional trees derived from the user-given  program rules.

\begin{figure}[!ht]
  \vspace*{-5mm}
  \scriptsize
  \begin{verbatim}
%% Data types:
type pages, id = int
type title, author, language, genre = [char]
data vocabularyLevel = easy | medium | difficult
data readerLevel = basic | intermediate | upper | proficiency
data book = book(id, title, author, language, genre, vocabularyLevel, pages)

%% Simple library, represented as list of books:
library :: [book]
library --> [ book(1, "Tintin", "Herge", "French", "Comic", easy, 65),
              book(2, "Dune", "F. P. Herbert", "English", "SciFi", medium, 345),
              book(3, "Kritik der reinen Vernunft", "Immanuel Kant", "German",
                   "Philosophy", difficult, 1011),
              book(4, "Beim Hauten der Zwiebel", "Gunter Grass", "German",
                   "Biography", medium, 432) ]

%% Auxiliary function for computing list membership:
member(B,[]) --> false
member(B,H:_T) --> true <== B == H
member(B,H:T) --> member(B,T) <== B /= H

%% Functions for getting the explicit attributes of a given book:
getId(book(Id,_Title,_Author,_Lang,_Genre,_VocLvl,_Pages)) --> Id
getTitle(book(_Id,Title,_Author,_Lang,_Genre,_VocLvl,_Pages)) --> Title
getAuthor(book(_Id,_Title,Author,_Lang,_Genre,_VocLvl,_Pages)) --> Author
getLanguage(book(_Id,_Title,_Author,Lang,_Genre,_VocLvl,_Pages)) --> Lang
getGenre(book(_Id,_Title,_Author,_Lang,Genre,_VocLvl,_Pages)) --> Genre
getVocabularyLevel(book(_Id,_Title,_Author,_Lang,_Genre,VocLvl,_Pages)) --> VocLvl
getPages(book(_Id,_Title,_Author,_Lang,_Genre,_VocLvl,Pages)) --> Pages

%% Function for guessing the genre of a given book:
guessGenre(B) --> getGenre(B)
guessGenre(B) -0.9-> "Fantasy" <== guessGenre(B) == "SciFi"
guessGenre(B) -0.8-> "Essay" <== guessGenre(B) == "Philosophy"
guessGenre(B) -0.7-> "Essay" <== guessGenre(B) == "Biography"
guessGenre(B) -0.7-> "Adventure" <== guessGenre(B) == "Fantasy"

%% Function for guessing the reader level of a given book:
guessReaderLevel(B) --> basic <== getVocabularyLevel(B) == easy, getPages(B) < 50
guessReaderLevel(B) -0.8-> intermediate <== getVocabularyLevel(B) == easy, getPages(B) >= 50
guessReaderLevel(B) -0.9-> basic <== guessGenre(B) == "Children"
guessReaderLevel(B) -0.9-> proficiency <== getVocabularyLevel(B) == difficult,
                                           getPages(B) >= 200
guessReaderLevel(B) -0.8-> upper <== getVocabularyLevel(B) == difficult, getPages(B) < 200
guessReaderLevel(B) -0.8-> intermediate <== getVocabularyLevel(B) == medium
guessReaderLevel(B) -0.7-> upper <== getVocabularyLevel(B) == medium

%% Function for answering a particular kind of user queries:
search(Language,Genre,Level) --> getId(B) <== member(B,library),
                                              getLanguage(B) == Language,
                                              guessReaderLevel(B) == Level,
                                              guessGenre(B) == Genre
\end{verbatim}
\normalsize\caption{Library with books in different languages}
\label{library}
\end{figure}

Figure \ref{library} shows a small  set of $\qcflp{\U}{\rdom}$ program rules, called the {\em library program} in the rest of the report. The concrete syntax is inspired by the functional logic language $\toy$, but the ideas and results of this report could be also applied to {\sf Curry} and other similar languages. In this example, $\U$ stands for a particular qualification domain which supports uncertain truth values in the real interval $[0,1]$, while $\rdom$ stands for a particular constraint domain which supports arithmetic constraints over the real numbers; see Section \ref{Qualification} for more details.

The program rules are intended to encode expert knowledge for computing qualified answers to user queries concerning the books available in a simplified library, represented as a list of objects of type \texttt{book}. The various \texttt{get} functions extract the explicit values of book attributes. Functions \texttt{guessGenre} and \texttt{guessReaderLevel} infer information by performing qualified inferences, relying on  analogies between different genres and heuristic rules to estimate reader levels on the basis of other features of a given book, respectively. Some program rules, as e.g. those of the auxiliary function \texttt{member}, have attached no explicit attenuation factor. By convention, this is understood as the implicit attachment of the attenuation factor \texttt{1.0}, the top value of $\U$. For any instance of the $\qcflp{\qdom}{\cdom}$ scheme, a similar convention allows to view $\cflp{\cdom}$-program rules as $\qcflp{\qdom}{\cdom}$-program rules whose attached qualification is optimal.

The last rule for function \texttt{search} encodes a method for computing qualified answers to a particular kind of user queries. Therefore, the queries can be formulated as goals to be solved by the program fragment. For instance, answering the query of a user who wants to find a book of genre {\tt "Essay"}, language {\tt "German"} and user level {\tt intermediate} with a certainty degree of at least 0.65 can be formulated as the goal:

\begin{center}\tt
   (search("German","Essay",intermediate) == R) \# W | W >= 0.65
\end{center}

\noindent The techniques presented in Section \ref{Implementing} can be used to translate the $\qcflp{\U}{\rdom}$ program rules and goal into the $\cflp{\rdom}$ language, which is implemented in the $\toy$ system. Solving the translated goal in $\toy$ computes the answer $\{R \mapsto 4\}\{0.65 \leq W, W \leq 0.7\}$, ensuring that the library book with \texttt{id} \texttt{4} satisfies the query's requirements with any certainty degree in the interval [0.65,0.7], in particular 0.7. The computation uses the 4th program rule of \texttt{guessGenre} to obtain \texttt{"Essay"} as the book's genre with qualification 0.7, and the 6th program rule of \texttt{guessReaderLevel} to obtain \texttt{intermediate} as the reader level with qualification 0.8.

The rest of the report is organized as follows. In Section \ref{Qualification} we recall known proposals concerning qualification  and constraint domains, and we introduce a technical notion needed to relate both kinds of domains for the purposes of this report. In Section \ref{QCFLP} we present the generic scheme $\qcflp{\qdom}{\cdom}$ announced in this introduction, and we formalize a special  Rewriting Logic which characterizes the declarative semantics of $\qcflp{\qdom}{\cdom}$-programs. In Section \ref{Implementing} we present a semantically correct program transformation converting $\qcflp{\qdom}{\cdom}$ programs and goals into the qualification-free $\cflp{\cdom}$ programming scheme, which is supported by existing systems such as $\toy$. Section \ref{Conclusions} concludes and points to some lines of planned future work.  
\section{Qualification and Constraint Domains} \label{Qualification}

\emph{Qualification Domains} were introduced in \cite{RR08}. Their  intended use has been already explained in the Introduction. In this section we recall and slightly improve their axiomatic definition.

\begin{definition}[Qualification Domains]
\label{def:qd}
A {\em Qualification Domain} is  any structure $\qdom = \langle D, \dleq, \bt, \tp, \circ \rangle$ verifying the following requirements:
\begin{enumerate}
  \item $D$, noted as $D_\qdom$ when convenient, is a set of elements called \emph{qualification values}.
    \item $\langle D, \dleq, \bt, \tp \rangle$ is a lattice with extreme points $\bt$ and $\tp$ w.r.t. the partial ordering $\dleq$. For given elements  $d, e \in D$, we  write $d \sqcap e$ for the {\em greatest lower bound} ($glb$) of $d$ and $e$, and $d \sqcup e$ for the {\em least upper bound} ($lub$) of $d$ and $e$. We also write $d \dlt e$ as abbreviation for $d \dleq e \land d \neq e$.
    \item $\circ : D \times D \longrightarrow D$, called {\em attenuation operation}, verifies the following axioms:
        \begin{enumerate}
            \item $\circ$ is associative, commutative and monotonic w.r.t. $\dleq$.
            \item $\forall d \in D : d \circ \tp = d$.
            \item $\forall d, e \in D \setminus \{\bt,\tp\} : d \circ e \dlt e$.
            \item $\forall d, e_1, e_2 \in D : d \circ (e_1 \sqcap e_2) = d \circ e_1 \sqcap d \circ e_2$. \qed
        \end{enumerate}
\end{enumerate}
\end{definition}

As an easy consequence of the previous definition one can prove the following proposition.
\footnote{The authors are thankful to G. Gerla for pointing out this fact.}

\begin{proposition}[Additional properties of qualification domains]
\label{prop:qd-properties}
Any qualification domain $\qdom$ satisfies the following properties:
\begin{enumerate}
  \item $\forall d,e \in D : d \circ e \dleq e$.
  \item $\forall d \in D : d \circ \bt = \bt$.
\end{enumerate}
\end{proposition}
\begin{proof}
Since $\tp$ is the top element of the lattice, we know $d \dleq \tp$ for any $d \in D$. As $\circ$ is monotonic w.r.t. $\dleq$, $d \circ e \dleq \tp \circ e$ also holds for any $e \in D$ which, due to commutativity and axiom ($b$) of $\circ$, yields $d \circ e \dleq e$. Therefore $\mathit{1.}$ holds. Now, taking $e = \bt$, one has $d \circ \bt \dleq \bt$ which implies $d \circ \bt = \bt$ as $\bt$ is the bottom element of the lattice. Hence $\mathit{2.}$ also holds. \qed
\end{proof}

The examples in this report will use a particular qualification domain $\U$ whose values represent certainty degrees in the sense of fuzzy logic. Formally,  $\U = \langle U, \leq, 0, 1,\times \rangle$, where $U = [0,1] = \{d \in \REAL \mid 0 \le d \le 1\}$, $\le$ is the usual numerical ordering, and $\times$ is the multiplication operation. In this domain,  the bottom and top elements are $\bt = 0$ and $\tp = 1$, and the infimum of a finite $S \subseteq U$ is the minimum value $\mbox{min}(S)$, understood as $1$ if $S = \emptyset$.
The class of qualification domains is closed under cartesian products. For a proof of this fact and other examples of qualification domains, the reader is referred to \cite{RR08,RR08TR}.

{\em Constraint domains} are used in Constraint Logic Programming and its extensions as a tool to provide data values, primitive operations and constraints tailored to domain-oriented applications. Various formalizations of this notion are known. In this report,  constraint domains are related to signatures of the form $\Sigma = \langle DC, PF, DF \rangle$ where $DC = \bigcup_{n \in \NAT}  DC^n$, $PF = \bigcup_{n \in \NAT}  PF^n$ and  $DF = \bigcup_{n \in \NAT}  DF^n$ are mutually disjoint sets of {\em data constructor} symbols, {\em primitive function} symbols and {\em defined function} symbols, respectively, ranked by arities. Given a signature $\Sigma$, a symbol $\bot$ to note the {\em undefined value}, a set $B$ of {\em basic values} $u$ and a countably infinite set $\Var$ of variables $X$, we define the notions listed below, where $\ntup{o}{n}$ abbreviates the $n$-tuple of syntactic objects $o_1, \ldots, o_n$.

\begin{itemize}
  \item {\em Expressions}  $e \in \Exp$ have the syntax $e ::= \bot | X | u | h(\ntup{e}{n})$, where $h \in DC^n \cup PF^n \cup DF^n$. In the case $n = 0$, $h(\ntup{e}{n})$ is written simply as $h$.
  \item {\em Constructor Terms}  $t \in \Term$ have the syntax $e ::= \bot | X | u | c(\ntup{t}{n})$,
where $c \in DC^n$. They will be called just terms in the sequel.
  \item {\em Total Expressions} $e \in \TExp$ and {\em Total Terms}  $t \in \TTerm$
have a similar syntax, with the $\bot$ case omitted.
  \item An expression or term (total or not) is called {\em ground} iff it includes no occurrences of variables.
$\GExp$ stands for the set of all ground expressions.
The notations $\GTerm$, $\TGExp$ and $\TGTerm$ have a similar meaning.
  \item We note as $\leqinfo$ the {\em information ordering}, defined as the least partial ordering over $\Exp$
compatible with contexts and verifying $\bot \leqinfo e$ for all $e \in \Exp$.
  \item Substitutions are defined as mappings $\sigma : \Var \to \Term$ assigning not necessarily total terms to variables. They can be represented as sets of bindings $X \mapsto t$ and extended to act over other syntactic objects $o$.
 The {\em domain} $\domset{\sigma}$ and {\em variable range} $\ranset{\sigma}$ are defined in the usual way.
  We will write $o \sigma$ for the result of applying $\sigma$ to $o$. The {\em composition} $\sigma\sigma'$ of two substitutions is such that $o(\sigma\sigma')$ equals $(o\sigma)\sigma'$.
\end{itemize}

By adapting the definition found in Section 2.2 of \cite{LRV07} to a first-order setting, we  obtain:
\footnote{We slightly modify  the statement of the {\em radicality} property, rendering it  simpler than in \cite{LRV07}  but sufficient for practical purposes.}

\begin{definition}[Constraint Domains]
\label{def:cd}
A {\em Constraint Domain} of signature $\Sigma$ is any algebraic structure of the form $\cdom = \langle C, \{p^\cdom \mid p \in PF\}\rangle$ such that:
\begin{enumerate}
  \item The carrier set $C$ is $\GTerm$ for a certain set $B$ of {\em basic values}. When convenient, we note $B$ and $C$ as $B_\cdom$ and $C_\cdom$, respectively.
  \item $p^{\cdom} \subseteq C^n \times C$, written simply as $p^{\cdom} \subseteq C$ in the case $n = 0$, is called the {\em interpretation} of $p$ in $\cdom$. We will write $p^{\cdom}(\ntup{t}{n}) \to t$ (or simply $p^{\cdom} \to t$ if $n = 0$) to indicate that $(\ntup{t}{n},\, t) \in p^{\cdom}$.
  \item Each primitive interpretation $p^{\cdom}$ has  {\em monotonic} and {\em radical} behavior w.r.t. the information ordering $\leqinfo$. More precisely:
\begin{enumerate}
\item {\bf Monotonicity}: For all $p \in PF^n$,  $p^{\cdom}(\ntup{t}{n}) \to t$ behaves
monotonically w.r.t. the arguments $\ntup{t}{n}$ and antimonotonically
w.r.t. the result $t$.
Formally: For all $\ntup{t}{n}, \ntup{t'}{n}, t, t' \in  C$ such that
$p^{\cdom}(\ntup{t}{n}) \to t$, $\ntup{t}{n} \sqsubseteq \ntup{t'}{n}$ and $t \sqsupseteq t'$,
$p^{\cdom}(\ntup{t'}{n}) \to t'$ also holds.
\item {\bf Radicality}: For all $p \in PF^n$, as soon as the arguments given to $p^{\cdom}$ have enough information to return  a result other than $\bot$, the same arguments suffice already  for returning a simple total result.
Formally:
For all $\overline{t}_n, t \in C$, if $p^{\cdom}(\overline{t}_n) \to t$ then $t=\bot$ or else $t \in B \cup DC^0$.
 \end{enumerate}
\end{enumerate}
\end{definition}

Note that symbols $h \in DC \cup DF$ are given no interpretation in $\cdom$. As we will see in Section \ref{QCFLP}, symbols in $c \in DC$ are interpreted as free constructors, and the interpretation of symbols $f \in DF$ is program-dependent. We assume that any signature $\Sigma$ includes two nullary constructors $true$ and $false$ for the boolean values, and a binary symbol $==\,\, \in PF^2$ used in infix notation and interpreted as {\em strict equality}; see \cite{LRV07} for details. For the examples in this report we will use a constraint domain $\rdom$ whose set of basic elements is $C_\rdom = \REAL$ and whose primitives functions correspond to the usual arithmetic operations $+,\times, \ldots$ and the usual boolean-valued comparison operations $\leq, <, \ldots$ over $\REAL$. Other useful instances of constraint domains can be found in \cite{LRV07}.

{\em Atomic constraints} over $\cdom$ have the form $p(\ntup{e}{n}) == v$
\footnote{Written as $p(\ntup{e}{n}) \to!\, v$ in \cite{LRV07}.}
 with $p \in PF^n$, $e_i \in \Exp$ and $v \in \Var \cup DC^0 \cup B_\cdom$. Atomic constraints of the form $p(\ntup{e}{n}) == true$ are abbreviated as $p(\ntup{e}{n})$. In particular, $(e_1 == e_2) == true$ is abbreviated as $e_1 == e_2$. Atomic constraints of the form $(e_1 == e_2) == false$ are abbreviated as $e_1 \cneq e_2$.

{\em Compound constraints} are built from  atomic constraints using logical conjunction, existential quantification, and sometimes other logical operations. Constraints without occurrences of symbols $f \in DF$ are called {\em primitive}. We will note atomic constraints as $\delta$, sets of atomic constraints as $\Delta$, atomic primitive constraints as $\pi$, and sets of atomic primitive constraints as $\Pi$. When interpreting set of constraints, we will treat them as the conjunction of their members.

Ground substitutions $\eta$ such that $X \eta \in \GTerm$ for all $X \in \domset{\eta}$ are called {\em variable valuations} over $\cdom$.
The set of all possible variable  valuations  is noted $\mbox{Val}_\cdom$. The {\em solution set} $\Solc{\Pi} \subseteq \mbox{Val}_\cdom$ includes as members those valuations $\eta$ such that $\pi\eta$ is true in $\cdom$ for all $\pi \in \Pi$; see \cite{LRV07} for a formal definition. In case that $\Solc{\Pi} = \emptyset$ we say that $\Pi$ is {\em unsatisfiable} and we write $\Unsat{\cdom}{\Pi}$. In case that $\Solc{\Pi} \subseteq \Solc{\pi}$ we say that $\pi$ is {\em entailed} by $\Pi$ in $\cdom$ and we write $\Pi \model{\cdom} \pi$. Note that the notions defined in this paragraph only make sense for primitive constraints.

In this report we are interested in pairs consisting of a qualification domain and a constraint domain that are related in the following technical sense:

\begin{definition} [Expressing $\qdom$ in $\cdom$]
\label{dfn:expressible}
A qualification domain $\qdom$ with carrier set $D_\qdom$ is expressible in a constraint domain $\cdom$ with carrier set $C_\cdom$  if $\aqdom \subseteq C_\cdom$ and the two following requirements are satisfied:
  \begin{enumerate}
     \item There is a primitive $\cdom$-constraint $\qval{X}$ depending on the variable $X$, such that $\Solc{\qval{X}} = \{\eta \in \mbox{Val}_\cdom \mid \eta(X) \in \aqdom\}$.
    \item There is a primitive $\cdom$-constraint $\qbound{X,Y,Z}$ depending on the variables $X$, $Y$, $Z$,
    such that any $\eta \in \mbox{Val}_\cdom$ such that $\eta(X)$, $\eta(Y)$, $\eta(Z) \in \aqdom$ verifies
    $\eta \in \Solc{\qbound{X,Y,Z}} \Longleftrightarrow \eta(X) \dleq \eta(Y) \circ \eta(Z)$. \qed
  \end{enumerate}
\end{definition}

Intuitively, $\qbound{X,Y,Z}$ encodes the $\qdom$-statement $X \dleq Y \circ Z$ as a $\cdom$-constraint. As convenient notations, we will write  $\encode{X \dleq Y \circ Z}$, $\encode{X \dleq Y}$ and $\encode{X \dgeq Y}$ in place of $\qbound{X,Y,Z}$, $\qbound{X,\tp,Y}$ and $\qbound{Y,\tp,X}$, respectively. In the sequel, $\cdom$-constraints of the form $\encode{\kappa}$ are called {\em qualification constraints}, and $\Omega$ is used as notation for sets of qualification constraints. We also write $\mbox{Val}_\qdom$ for the set  of all $\mu \in \mbox{Val}_\cdom$ such that $X\mu \in \aqdom$ for all $X \in \domset{\mu}$, called {\em $\qdom$-valuations}.

Note that $\U$ can be expressed in $\rdom$, because $D_\U \setminus \{0\} = (0,1] \subseteq \REAL \subseteq C_\rdom$, $\qval{X}$ can be built as the $\rdom$-constraint $0 < X \land X \leq 1$ and $\encode{X \dleq Y \circ Z}$ can be built as the $\rdom$-constraint $X \leq Y \times Z$. Other instances of qualification domains presented in \cite{RR08} are also expressible in $\rdom$.
\section{A Qualified Declarative Programming Scheme} \label{QCFLP}

In this section we present the scheme $\qcflp{\qdom}{\cdom}$ announced in the Introduction, and we develop alternative characterizations of its declarative semantics using an interpretation transformer and a rewriting logic. The parameters $\qdom$ and $\cdom$  respectively stand for a qualification domain and a constraint domain with certain signature $\Sigma$. By convention, we only allow those instances of the scheme verifying that $\qdom$ is expressible in $\cdom$ in the sense of Definition \ref{dfn:expressible}. For example, $\qcflp{\U}{\rdom}$ is an allowed instance.

Technically, the results presented here extend similar ones known for the $\cflp{\cdom}$ sheme  \cite{LRV07}, omitting higher-order functions and adding a suitable treatment of qualifications. In particular, the qc-interpretations for $\qcflp{\qdom}{\cdom}$-programs are a natural extension of the c-interpretations for $\cflp{\cdom}$-programs introduced in \cite{LRV07}. In turn, these were inspired by the $\pi$-interpretations for the $\clp{\cdom}$ scheme proposed by Dore, Gabbrielli and Levi \cite{GL91,GDL95}.

\subsection{Programs, Interpretations and Models}

A $\qcflp{\qdom}{\cdom}$-program is a set $\Prog$ of program rules. A program rule has the form $f(\ntup{t}{n}) \qto{\alpha} r \Leftarrow \Delta$ where $f \in DF^n$, $\ntup{t}{n}$ is a  lineal sequence of $\Sigma$-terms, $\alpha \in \aqdom$ is an attenuation factor,  $r$ is a $\Sigma$-expression and $\Delta$ is a sequence of atomic $\cdom$-constraints $\delta_j ~ (1 \leq j \leq m)$,  interpreted as conjunction. The undefined symbol $\bot$ is not allowed to occur in program rules.

The library program shown in Figure \ref{library} is  an example of $\qcflp{\U}{\rdom}$-program. Leaving aside the attenuation factors, this is clearly not a confluent conditional term rewriting system. Certain program rules, as e.g. those for \texttt{guessGenre}, are intended to specify the behavior of {\em non-deterministic functions}. As argued elsewhere \cite{Rod01}, the semantics of non-deterministic functions for the purposes of Functional Logic Programming is not suitably described by ordinary rewriting.
Inspired by the approach in \cite{LRV07},  we will overcome this difficulty by designing special inference mechanisms to derive semantically meaningful statements from programs. The kind of statements that  we will consider are defined below:

\begin{definition}[qc-Statements]
\label{dfn:qc-statements}
Assume partial $\Sigma$-expression $e$,  partial $\Sigma$-terms $t, t', \ntup{t}{n}$, a qualification value $d \in \aqdom$, an atomic $\cdom$-constraint $\delta$ and a finite set of atomic primitive $\cdom$-constraints $\Pi$. A {\em qualified constrained statement} (briefly, qc-statement) $\varphi$ must have one of the following two forms:
       \begin{enumerate}
        \item \emph{qc-production} $(e \to t)\sharp d \Leftarrow \Pi$. Such a qc-statement is called \emph{trivial} iff either $t$ is $\bot$ or else $\Unsat{\cdom}{\Pi}$. Its intuitive meaning is that a rewrite sequence $e \to^* t'$ using program rules and with attached qualification value $d$ is allowed in our intended semantics for some $t' \sqsupseteq t$,
under the assumption that $\Pi$ holds.
By convention, qc-productions of the form  $(f(\ntup{t}{n}) \to t)\sharp d \Leftarrow \Pi$ with $f \in DF^n$
are called \emph{qc-facts}.
        \item \emph{qc-atom} $\delta\sharp d \Leftarrow \Pi$. Such a qc-statement is called \emph{trivial} iff $\Unsat{\cdom}{\Pi}$. Its intuitive meaning is that $\delta$ is entailed by the program rules with attached qualification value $d$, under the assumption that $\Pi$ holds. \qed
      \end{enumerate}
\end{definition}

Our semantics will use program interpretations defined as sets of qc-facts with certain closure properties. As an auxiliary tool we need the following technical notion:

\begin{definition}[$(\qdom,\cdom)$-Entailment]
\label{dfn:entailment}
Given two qc-statements $\varphi$ and $\varphi'$, we say that $\varphi$ $(\qdom,\cdom)$-entails $\varphi'$ (in symbols, $\varphi \entail{\qdom,\cdom} \varphi'$) iff one of the following two cases hold:
      \begin{enumerate}
        \item $\varphi = (e \to t)\sharp d \Leftarrow \Pi$, $\varphi' = (e' \to t')\sharp d' \Leftarrow \Pi'$, and there is some substitution $\sigma$ such that $\Pi' \model{\cdom} \Pi\sigma$, $d \dgeq d'$, $e\sigma \sqsubseteq e'$ and $t\sigma \sqsupseteq t'$.
        \item $\varphi = \delta\sharp d \Leftarrow \Pi$, $\varphi' = \delta'\sharp d' \Leftarrow \Pi'$, and there is some substitution $\sigma$ such that $\Pi' \model{\cdom} \Pi\sigma$, $d \dgeq d'$, $\delta\sigma \sqsubseteq \delta'$. \qed
      \end{enumerate}
\end{definition}

The intended meaning of $\varphi \entail{\qdom,\cdom} \varphi'$ is that $\varphi'$ follows from  $\varphi$, regardless of the interpretation of the defined function symbols $f \in DF$ occurring in $\varphi$, $\varphi'$. Intuitively, this is the case because the interpretations of defined function symbols are expected to satisfy the monotonicity properties stated for the case of primitive function symbols in Definition \ref{def:cd}. The following example may help to understand the idea:

\begin{example}[$(\U,\rdom)$-entailment]
\label{exmp:entailment}
Let $\varphi$, $\varphi'$ be defined as:
$$
\begin{array}{lll}
\varphi &:& (f(X\!:\!Xs) \rightarrow Xs)\sharp 0.8 \Leftarrow X\times X \neq 0 \\
\varphi' &:& (f(A\!:\!(B\!:\![\,])) \rightarrow \bot\!:\!\bot)\sharp 0.7 \Leftarrow A < 0
\end{array}
$$

\noindent Then $\varphi \entail{\U,\rdom} \varphi'$ with $\sigma = \{ X \mapsto A,\ Xs \mapsto B\!:\!\bot\}$ because:
\begin{itemize}
\item $\Pi' \model{\rdom} \Pi\sigma$, since $\Pi' = \{ A < 0 \}$, $\Pi\sigma= \{ X\times X \neq 0 \}\sigma = \{ A\times A \neq 0 \}$,
and $A \times A \neq 0$ is entailed by $A < 0$ in $\rdom$.
\item $d \dgeq d'$ holds  in $\U$, since $d = 0.8 \geq 0.7 = d'$.
\item $e\sigma \sqsubseteq e'$, since $e\sigma = f(X\!:\!Xs)\sigma =  f(A\!:\!(B\!:\!\bot)) \sqsubseteq f(A\!:\!(B\!:\![\,])) = e'$.
\item $t\sigma \sqsupseteq t'$, since $t\sigma = Xs\sigma =  B\!:\!\bot \sqsupseteq \bot\!:\!\bot = t'$. \qed
\end{itemize}
\end{example}

Now we can define program interpretations as follows:

\begin{definition}[qc-Interpretations] 
\label{dfn:interpretation}
A \emph{qualified constrained interpretation} (or \emph{qc-interpretation}) over $\qdom$ and $\cdom$ is a set $\I$ of qc-facts including all trivial and entailed qc-facts. In other words, a set $\I$ of qc-facts such that $\closure{\qdom,\cdom}{\I} \subseteq \I$, where the closure over $\qdom$ and $\cdom$ of $\I$ is defined as: $$\closure{\qdom,\cdom}{\I} \eqdef \{ \varphi \mid \varphi \text{ trivial}\} \cup \{\varphi' \mid \varphi \entail{\qdom,\cdom} \varphi' \text{ for some } \varphi \in \I\} \enspace .$$ We write $\Intdc$ for the set of all qc-interpretations over $\qdom$ and $\cdom$.
\end{definition}

\begin{figure}[ht]
  \small
  \centering
    \begin{tabular}{|@{\hspace*{0.5cm}}l@{\hspace*{0.5cm}}|}\hline
      \\
      \textbf{QTI} $\quad$ $\displaystyle\frac{}{\quad\varphi\quad}$ ~
      if $\varphi$ is a trivial qc-statement.\\
      \\
      \textbf{QRR} $\quad$ $\displaystyle\frac{}{(v \to v)\sharp d \Leftarrow \Pi}$ ~
      if $v \in \Var \cup B_\cdom$ and $d \in \aqdom$.\\
      \\
      \textbf{QDC} $\quad$ $\displaystyle\frac
        {(~ (e_i \to t_i)\sharp d_i \Leftarrow \Pi ~)_{i = 1 \ldots n}}
        {(c(\ntup{e}{n}) \to c(\ntup{t}{n}))\sharp d \Leftarrow \Pi}$ ~
      \begin{tabular}{l}
        if $c \in DC^n$ and $d \in \aqdom$ \\
        verifies $d \dleq d_i$ $(1 \le i \le n)$.\\
      \end{tabular} \\
      \\
      \textbf{QDF$_{\I}$} $\quad$ $\displaystyle\frac
        {(~ (e_i \to t_i)\sharp d_i \Leftarrow \Pi ~)_{i = 1 \ldots n}}
        {(f(\ntup{e}{n}) \to t)\sharp d \Leftarrow \Pi}$ \\ \\
      if $f \in DF^n$, non-trivial $((f(\ntup{t}{n}) \to t) \sharp d_0 \Leftarrow \Pi) \in \I$ \\
      and $d \in \aqdom$ verifies $d \dleq d_i ~ (0 \le i \le n)$.\\
      \\
      \textbf{QPF} $\quad$ $\displaystyle\frac
        {(~ (e_i \to t_i)\sharp d_i \Leftarrow \Pi ~)_{i=1 \ldots n}}
        {(p(\ntup{e}{n}) \to v)\sharp d \Leftarrow \Pi}$ ~
      if $p \in PF^n$, $v \in \Var \cup DC^0 \cup B_\cdom$, \\ \\
      $\Pi \model{\cdom} p(\ntup{t}{n}) \to v$ and $d \in \aqdom$ verifies $d \dleq d_i ~ (1 \le i \le n)$.\\
      \\
      \textbf{QAC} $\quad$ $\displaystyle\frac
        {(~ (e_i \to t_i)\sharp d_i \Leftarrow \Pi ~)_{i=1 \ldots n}}
        {(p(\ntup{e}{n}) == v)\sharp d \Leftarrow \Pi}$ ~
      if $p \in PF^n$, $v \in \Var \cup DC^0 \cup B_\cdom$, \\ \\
      $\Pi \model{\cdom} p(\ntup{t}{n}) == v$ and $d \in \aqdom$ verifies $d \dleq d_i ~ (1 \le i \le n)$.\\
      \\\hline
    \end{tabular}
  \caption{Qualified Constrained Rewriting Logic for Interpretations}
  \label{fig:iqcrwl}
\end{figure}

Given a qc-interpretation $\I$, the inference rules displayed in Fig. \ref{fig:iqcrwl} are used to derive qc-statements from the qc-facts belonging to $\I$. The inference system consisting of these rules is called {\em Qualified Constrained Rewriting Logic for Interpretations} and noted as $\IQCRWLdc$. The notation $\I \iqcrwldc \varphi$ is used to indicate that $\varphi$ can be derived from $\I$ in  $\IQCRWLdc$. By convention, we agree that no other  inference rule is used whenever \textbf{QTI} is applicable. Therefore, trivial qc-statements can only be inferred by rule \textbf{QTI}. As usual in formal inference systems, $\IQCRWLdc$ proofs can be represented as trees whose nodes correspond to inference steps.

In the sequel, the inference rules \textbf{QDF}$_\I$, \textbf{QPF} and \textbf{QAC} will be called {\em crucial}. The notation $|\mathcal{T}|$  will denote the number of inference steps within the proof tree $\mathcal{T}$ that are {\em not crucial}.  Proof trees with no crucial inferences (i.e. such that $|\mathcal{T}| = 0$) will be called {\em easy}. The following lemma states some technical properties of $\IQCRWLdc$.

\begin{lemma}[Some properties of $\IQCRWLdc$]
\label{lem:iqcrwl-properties}
\begin{enumerate}
  \item\label{lem:iqcrwl-properties:1} Approximation property: For any non-trivial $\varphi$ of the form $(t \to t')\sharp d \Leftarrow \Pi$ where $t, t' \in \Term$, the three following affirmations are equivalent:
  (a) $t \sqsupseteq t'$;
  (b) $\I \iqcrwldc \varphi$ with an easy proof tree; and
  (c) $\I \iqcrwldc \varphi$.
  \item\label{lem:iqcrwl-properties:2} Primitive c-atoms: For any primitive c-atom $p(\ntup{t}{n}) == v$, one has $\I \iqcrwldc (p(\ntup{t}{n}) == v)\sharp d \Leftarrow \Pi \iff  \Pi \model{\cdom} p(\ntup{t}{n}) == v$.
  \item\label{lem:iqcrwl-properties:3} Entailment property: $\I \iqcrwldc \varphi$  with a proof tree $\mathcal{T}$
  and $\varphi \entail{\qdom,\cdom} \varphi'$ $\Longrightarrow$ $\I \iqcrwldc \varphi'$ with a proof tree $\mathcal{T'}$ such that  $|\mathcal{T'}| \leq |\mathcal{T}|$.
  \item\label{lem:iqcrwl-properties:4} Conservation property: For any qc-fact $\varphi$, one has $\I \iqcrwldc \varphi \iff \varphi \in \I$.
\end{enumerate}
\end{lemma}
\begin{proof}
We argue separately for each of the four properties:

\smallskip
\noindent [$\mathit{\ref{lem:iqcrwl-properties:1}.}$] (\emph{Approximation property}).
The terms $t$ and $t'$ involve neither defined nor primitive function symbols.
Due to the form of the $\IQCRWLdc$ inference rules, a proof of the qc-statement $(t \to t')\sharp d \Leftarrow \Pi$ will involve no crucial inferences and it will succeed iff $t \sqsupseteq t'$.
A formal proof can be easily obtained reasoning by induction on the syntactic size of $t$, similarly as in  item 3. of Lemma 1 from \cite{LRV07}.

\smallskip
\noindent [$\mathit{\ref{lem:iqcrwl-properties:2}.}$] (\emph{Primitive c-atoms}).
Let $\varphi$ be $(p(\ntup{t}{n}) == v)\sharp d \Leftarrow \Pi$. If $\varphi$ is trivial,
then $\I \iqcrwldc (p(\ntup{t}{n}) == v)\sharp d \Leftarrow \Pi$ can be proved with just one {\bf QTI} inference,
and $\Pi \model{\cdom} p(\ntup{t}{n}) == v$ also holds because of $\Unsat{\cdom}{\Pi}$.
If $\varphi$ is not trivial, then:

  \begin{itemize}
    \item ($\Leftarrow$) Assume $\Pi \model{\cdom} p(\ntup{t}{n}) == v$.
    Then $\I \iqcrwldc (p(\ntup{t}{n}) == v)\sharp d \Leftarrow \Pi$ can be obtained with a proof of the form
      $$
      \displaystyle\frac
      {(~ (t_i \to t_i)\sharp \tp \Leftarrow \Pi ~)_{i=1 \ldots n}}
      {(p(\ntup{t}{n}) == v)\sharp d \Leftarrow \Pi} \text{ QAC}
      $$
      where each of the $n$ premises has an easy $\IQCRWLdc$-proof  due to the approximation property
      (since $t_i \sqsupseteq t_i$).
      \item ($\Rightarrow$) Assume now $\I \iqcrwldc (p(\ntup{t}{n}) == v)\sharp d \Leftarrow \Pi$. The $\IQCRWLdc$-proof
      will have the form
      $$
      \displaystyle\frac
      {(~ (t_i \to t'_i)\sharp d_i \Leftarrow \Pi ~)_{i=1 \ldots n}}
      {(p(\ntup{t}{n}) == v)\sharp d \Leftarrow \Pi} \text{ QAC}
      $$
      where $\Pi \model{\cdom} p(\ntup{t'}{n}) == v$ and
      $\I \iqcrwldc (t_i \to t_i)\sharp d_i \Leftarrow \Pi$, $d \dleq d_i$ hold for all $1 \leq i \leq n$.
      Due to the approximation property, we can conclude that $t_i \sqsupseteq t'_i$ holds for $1 \leq i \leq n$,
      which implies  $\Pi \model{\cdom} p(\ntup{t}{n}) == v$  because of  the monotonic behavior of primitive functions in constraint domains.
  \end{itemize}

\smallskip
\noindent [$\mathit{\ref{lem:iqcrwl-properties:3}.}$] (\emph{Entailment property}).
Assume $\I \iqcrwldc \varphi$  with a $\IQCRWLdc$-proof tree $\mathcal{T}$.
We must prove that $\I \iqcrwldc \varphi'$ with some  proof tree $\mathcal{T'}$ such that  $|\mathcal{T'}| \leq |\mathcal{T}|$.
If $\varphi'$ results trivial, then it is proved with just one \textbf{QTI} inference step,
and therefore $|\mathcal{T'}| = 0 \leq |\mathcal{T}|$.
In the sequel, we assume $\varphi'$ non-trivial and we reason by induction on the number of inference steps within $\mathcal{T}$.  We distinguish cases according to the inference step at the root of $\mathcal{T}$:
 \begin{itemize}
    \item \textbf{QTI}: From Definition \ref{dfn:entailment} it is easy to check that $\varphi'$ must be trivial whenever $\varphi \entail{\qdom,\cdom} \varphi'$ and $\varphi$ is trivial. Since we are assuming that $\varphi'$ is not trivial, this case cannot happen.
    \item \textbf{QRR}: In this case $\varphi$ is of the form $(v \to v)\sharp d \Leftarrow \Pi$ with either $v \in B_\cdom$ or $v \in \Var$. Since $\varphi \entail{\qdom,\cdom} \varphi'$, we assume $\varphi' : (v' \to v')\sharp d' \Leftarrow \Pi'$ with $\Pi' \model{\cdom} \Pi\sigma$, $d \dgeq d'$ and $v\sigma = v'$ for some substitution $\sigma$. If $v \in B_\cdom$, then also $v' \in B_\cdom$ and $\I \iqcrwldc \varphi'$ can be proved with a proof tree $\mathcal{T'}$ consisting of just one \textbf{QRR} inference step. If $v \in \Var$, then $v' \in \Term$, and $\I \iqcrwldc \varphi'$ can be  proved with a proof tree $\mathcal{T'}$ consisting only of \textbf{QDC} and \textbf{QRR} inferences. In both cases, $|\mathcal{T'}| = 0 \leq |\mathcal{T}|$.

    \item \textbf{QDC}: In this case $\varphi : (c(\ntup{e}{n}) \to c(\ntup{t}{n}))\sharp d \Leftarrow \Pi$ and $\mathcal{T}$ has the form $$\displaystyle\frac{(~ (e_i \to t_i)\sharp d_i \Leftarrow \Pi ~)_{i=1 \ldots n}}{(c(\ntup{e}{n}) \to c(\ntup{t}{n}))\sharp d \Leftarrow \Pi}  \text{ QDC}$$ where $c \in DC^n$,
$\I \iqcrwldc (e_i \to t_i)\sharp d_i \Leftarrow \Pi$ with proof tree $\mathcal{T}_i$,
 and $d \dleq d_i ~ (1 \le i \le n)$. Since $\varphi \entail{\qdom,\cdom} \varphi'$, we can assume that $\varphi'$ has the form $(c(\ntup{e'}{n}) \to c(\ntup{t'}{n}))\sharp d' \Leftarrow \Pi'$ with $e_i\sigma \sqsubseteq e'_i ~ (1 \le i \le n)$, $c(\ntup{t}{n})\sigma \sqsupseteq c(\ntup{t'}{n})$, $d \dgeq d'$ and $\Pi' \model{\cdom} \Pi\sigma$ for some substitution $\sigma$. For $1 \leq i \leq n$, we clearly obtain
 $(e_i \to t_i)\sharp d_i \Leftarrow \Pi \entail{\qdom,\cdom} (e'_i \to t'_i)\sharp d_i \Leftarrow \Pi'$,
 and by induction hypothesis we can assume $\I \iqcrwldc (e'_i \to t'_i)\sharp d_i \Leftarrow \Pi'$
 with proof tree $\mathcal{T'}_i$ such that $|\mathcal{T'}_i| \leq |\mathcal{T}_i|$.
 Then we get $\I  \iqcrwldc (c(\ntup{e'}{n}) \to c(\ntup{t'}{n}))\sharp d' \Leftarrow \Pi'$ with a
 proof tree $\mathcal{T'}$ such that $|\mathcal{T'}| \leq |\mathcal{T}|$.
More precisely,  $\mathcal{T'}$ has the form
$$\displaystyle\frac{(~ (e'_i \to t'_i)\sharp d_i \Leftarrow \Pi' ~)_{i=1 \ldots n}}{(c(\ntup{e'}{n}) \to c(\ntup{t'}{n}))\sharp d' \Leftarrow \Pi'}  \text{ QDC}$$ where $d' \dleq d_i$ follows from $d' \dleq d \dleq d_i ~ (1 \le i \le n)$ and each
premise is proved by $\mathcal{T'}_i$.

    \item \textbf{QDF}$_\I$: In this case   $\varphi : (f(\ntup{e}{n}) \to t)\sharp d \Leftarrow \Pi$ and $\mathcal{T}$ has the form $$\displaystyle\frac{(~ (e_i \to t_i)\sharp d_i \Leftarrow \Pi ~)_{i=1 \ldots n}}{(f(\ntup{e}{n}) \to t)\sharp d \Leftarrow \Pi} \text{ QDF}_\I$$ where $f \in DF^n$ and there is some non-trivial $\psi = (f(\ntup{t}{n}) \to t)\sharp d_0 \Leftarrow \Pi)$ such that $\psi \in \I$, $\I \iqcrwldc (e_i \to t_i)\sharp d_i \Leftarrow \Pi$ with proof tree $\mathcal{T}_i$ and $d \dleq d_i ~ (0 \le i \le n)$. Since $\varphi \entail{\qdom,\cdom} \varphi'$, we can assume $\varphi' = (f(\ntup{e'}{n}) \to t')\sharp d' \Leftarrow \Pi'$ with $e_i\sigma \sqsubseteq e'_i ~ (1 \le i \le n)$, $t\sigma \sqsupseteq t'$, $d \dgeq d'$ and $\Pi' \model{\cdom} \Pi\sigma$ for some substitution $\sigma$.
For $1 \leq i \leq n$, we get
$(e_i \to t_i)\sharp d_i \Leftarrow \Pi \entail{\qdom,\cdom} (e'_i \to t_i\sigma)\sharp d_i \Leftarrow \Pi'$,
and by induction hypothesis we can assume
$\I \iqcrwldc (e'_i \to t_i\sigma)\sharp d_i \Leftarrow \Pi'$
with proof tree $\mathcal{T'}_i$ such that $|\mathcal{T'}_i| \leq |\mathcal{T}_i|$.
Consider now $\psi' = ((f(\ntup{t}{n})\sigma \to t')\sharp d_0 \Leftarrow \Pi')$.
Clearly, $\psi \entail{\qdom,\cdom} \psi'$ and therefore $\psi' \in \I$ because $\I$ is closed under $(\qdom, \cdom)$-entailment. Using this $\psi'$ we get $\I  \iqcrwldc (f(\ntup{e'}{n}) \to t')\sharp d' \Leftarrow \Pi'$ with a
 proof tree $\mathcal{T'}$ such that $|\mathcal{T'}| \leq |\mathcal{T}|$.
More precisely,  $\mathcal{T'}$ has the form
$$\displaystyle\frac{(~ (e'_i \to t_i\sigma)\sharp d_i \Leftarrow \Pi' ~)_{i=1 \ldots n}}{(f(\ntup{e'}{n}) \to t')\sharp d' \Leftarrow \Pi'} \text{ QDF}_\I$$
where $d' \dleq d_i$ follows from $d' \dleq d \dleq d_i ~ (0 \le i \le n)$ and each
premise is proved by $\mathcal{T'}_i$.

    \item \textbf{QPF}: In this case $\varphi : (p(\ntup{e}{n}) \to v)\sharp d \Leftarrow \Pi$ and $\mathcal{T}$ has the form $$\displaystyle\frac{(~ (e_i \to t_i)\sharp d_i \Leftarrow \Pi ~)_{i=1 \ldots n}}{(p(\ntup{e}{n}) \to v)\sharp d \Leftarrow \Pi} \text{ QPF}$$ where $p \in PF^n$, $v \in \Var \cup DC^0 \cup B_\cdom$, $\Pi \model{\cdom} p(\ntup{t}{n}) \to v$,
$d \dleq d_i $ and $\I \iqcrwldc (e_i \to t_i)\sharp d_i \Leftarrow \Pi$ with proof tree $\mathcal{T}_i\, (1 \leq i \leq n)$.
Since  $\varphi \entail{\qdom,\cdom} \varphi'$, we can assume $\varphi'$ to be of the form $(p(\ntup{e'}{n}) \to v')\sharp d' \Leftarrow \Pi'$ with $e_i\sigma \sqsubseteq e'_i ~ (1 \le i \le n)$, $v\sigma \sqsupseteq v'$, $d \dgeq d'$ and $\Pi' \model{\cdom} \Pi\sigma$ for some substitution $\sigma$.
For $1 \leq i \leq n$, we get
$(e_i \to t_i)\sharp d_i \Leftarrow \Pi \entail{\qdom,\cdom} (e'_i \to t_i\sigma)\sharp d_i \Leftarrow \Pi'$,
and by induction hypothesis we can assume
$\I \iqcrwldc (e'_i \to t_i\sigma)\sharp d_i \Leftarrow \Pi'$
with proof tree $\mathcal{T'}_i$ such that $|\mathcal{T'}_i| \leq |\mathcal{T}_i|$.
Moreover, we can also assume
$v' \in \Var \cup DC^0 \cup B_\cdom$ because $p$ is a primitive function symbol and $\varphi'$ is not trivial.
From $v, v' \in \Var \cup DC^0 \cup B_\cdom$ and $v\sigma \sqsupseteq v'$ we can conclude that $v\sigma = v'$.
Then,  from $\Pi \model{\cdom} p(\ntup{t}{n}) \to v$ and $\Pi' \model{\cdom} \Pi\sigma$ we can deduce
$\Pi' \model{\cdom} p(\ntup{t}{n})\sigma \to v'$.
Putting everything together, we get
$\I  \iqcrwldc (p(\ntup{e'}{n}) \to v')\sharp d' \Leftarrow \Pi'$
with a proof tree $\mathcal{T'}$ such that $|\mathcal{T'}| \leq |\mathcal{T}|$.
More precisely,  $\mathcal{T'}$ has the form
$$\displaystyle\frac{(~ (e'_i \to t_i\sigma)\sharp d_i \Leftarrow \Pi' ~)_{i=1 \ldots n}}{(p(\ntup{e'}{n}) \to v')\sharp d' \Leftarrow \Pi'} \text{ QPF}$$
where $d' \dleq d_i$ follows from $d' \dleq d \dleq d_i ~ (1 \le i \le n)$ and each
premise is proved by $\mathcal{T'}_i$.

    \item \textbf{QAC}: Similar to the case for \textbf{QPF}.
  \end{itemize}

\smallskip
\noindent [$\mathit{\ref{lem:iqcrwl-properties:4}.}$] (\emph{Conservation property}).
Assume $\varphi : (f(\ntup{t}{n}) \to t)\sharp d \Leftarrow \Pi$. In the case that $\varphi$ is a trivial qc-fact, it is true by definition of qc-interpretation that $\varphi \in \I$, and $\I \iqcrwldc \varphi$ follows by rule \textbf{QTI}. Therefore the property is satisfied for trivial qc-facts. If $\varphi$ is not trivial, we prove each implication as follows:
  \begin{itemize}
    \item ($\Leftarrow$) Assume $\varphi \in \I$. Then $\I \iqcrwldc \varphi$ with a $\IQCRWLdc$-proof tree of the form:
      $$
      \displaystyle\frac
      {(~ (t_i \to t_i)\sharp \tp \Leftarrow \Pi ~)_{i=1 \ldots n}}
      {(f(\ntup{t}{n}) \to t)\sharp d \Leftarrow \Pi} \text{ QDF}_\I \text{ using } \varphi \in \I
      $$ where each premise has an easy  $\IQCRWLdc$-proof tree due to the approximation property, and  $d \dleq d, \tp$ hold trivially.
    \item ($\Rightarrow$) Assume $\I \iqcrwldc \varphi$. As $\varphi$ is not trivial, there is a $\IQCRWLdc$-proof tree of the form:
      $$
      \displaystyle\frac
      {(~ (t_i \to t'_i)\sharp d_i \Leftarrow \Pi ~)_{i=1 \ldots n}}
      {(f(\ntup{t}{n}) \to t)\sharp d \Leftarrow \Pi} \text{ QDF}_\I
      \text{ using } \varphi' = (f(\ntup{t'}{n}) \to t) \sharp d' \Leftarrow \Pi) \in \I
      $$
      where $d \dleq d', d_i$ and  $\I \iqcrwldc (t_i \to t'_i)\sharp d_i \Leftarrow \Pi\, (1 \le i \le n)$.
      For each $1 \leq i \leq n$, we claim that $t'_i \sqsubseteq t_i$.
      If $t'_i = \bot$ the claim is trivial.
      If $t'_i \neq \bot$, then $(t_i \to t'_i)\sharp d_i \Leftarrow \Pi$ is a non-trivial qc-production
      and the claim follows from $\I \iqcrwldc (t_i \to t'_i)\sharp d_i \Leftarrow \Pi$
      and the approximation property.
      Now, the claim together with $\Pi \model{\cdom} \Pi$, $d' \dgeq d$ and $t \sqsupseteq t$
      yields $\varphi' \entail{\qdom,\cdom} \varphi$.
      Since $\varphi' \in \I$ and $\I$ is closed under $(\qdom, \cdom)$-entailment,
      we can conclude that $\varphi \in \I$. \qed
       \end{itemize}
\end{proof}

Next, we can define program models and semantic consequence, adapting ideas from the so-called {\em strong semantics} of \cite{LRV07}.
\footnote{Weak models and weak semantic consequence could be also defined similarly as in \cite{LRV07}, but strong semantics suffices for the purposes of this report.}

\begin{definition}[Models and semantic consequence] 
\label{dfn:model}
Let a $\qcflp{\qdom}{\cdom}$-program $\Prog$  be given.
\begin{enumerate}
\item
A qc-interpretation $\I$ is a model of  $R_l : (f(\ntup{t}{n}) \qto{\alpha} r \Leftarrow \ntup{\delta}{m}) \in \Prog$
(in symbols, $\I \model{\qdom,\cdom} R_l$) iff
for every  substitution $\theta$,
for every set of atomic primitive $\cdom$-constraints $\Pi$,
for every c-term $t \in \Term$ and for all $d, d_0,\ldots,d_m \in \aqdom$ such that
$\I \iqcrwldc \delta_i\theta\sharp d'_i \Leftarrow \Pi\, (1 \le i \le m)$,
$\I \iqcrwldc (r\theta \to t)\sharp d'_0 \Leftarrow \Pi$ and
$d \dleq \alpha \circ d_i\, (0 \le i \le m)$, one has
$((f(\ntup{t}{n})\theta \to t)\sharp d \Leftarrow \Pi) \in \I$.
\item
A qc-interpretation $\I$ is a \emph{model} of $\Prog$ (in symbols, $\I \model{\qdom,\cdom} \Prog$) iff
$\I$ is a model of every program rule belonging to $\Prog$.
\item
A qc-statement $\varphi$ is a semantic consequence of $\Prog$ (in symbols, $\Prog \model{\qdom,\cdom} \varphi$) iff
$\I \iqcrwldc \varphi$ holds for every qc-interpretation $\I$ such that  $\I \model{\qdom,\cdom} \Prog$. \qed
\end{enumerate}
\end{definition}

\subsection{Least Models}

We  will now present two different characterizations for the least model of a given program $\Prog$: in the first place as a least fixpoint of an interpretation transformer
and in the second place as the set of qc-facts derivable from $\Prog$ in a special rewriting logic.

\subsection*{\it A fixpoint characterization of least models.}

A well-known way of characterizing least program models is to exploit the lattice structure of the family of all program interpretations to obtain the least model of a given program $\Prog$ as the least fixpoint of an interpretation transformer related to $\Prog$. Such characterizations are know for logic programming \cite{Llo87,Apt90}, constraint logic programming \cite{GL91,GDL95,JMM+98}, constraint functional logic programming \cite{LRV07} and qualified logic programming \cite{RR08}.
Our approach here extends that in \cite{LRV07} by adding qualification values.

The next result, whose easy proof is omitted, provides a lattice structure of program interpretations:

\begin{proposition}[Interpretations Lattice]
\label{prop:interpretation-lattice}
$\Intdc$ defined as the set of all qc-interpretations over the qualification domain $\qdom$ and the constraint domain $\cdom$ is a complete lattice w.r.t. the set inclusion ordering ($\subseteq$). Moreover, the bottom element $\ibot$ and the top element $\itop$ of this lattice are characterized as $\ibot = \closure{\qdom,\cdom}{\{\varphi \mid \varphi \mbox{ is a trivial qc-fact}\}}$ and $\itop = \{\varphi \mid \varphi \mbox{ is any qc-fact}\}$.
\end{proposition}

 Now we define an {\em interpretations transformer} $\STp$ intended to formalize the computation of immediate consequences from the qc-facts belonging to a given qc-interpretation.

 \begin{definition}[Interpretations transformers] 
\label{dfn:transformers}
Assuming a $\qcflp{\qdom}{\cdom}$-program $\Prog$ and a qc-interpreta\-tion $\I$, $\STp : \Intdc \to \Intdc$ is defined as $\STp(\I) \eqdef \closure{\qdom,\cdom}{pre\STp(\I)}$ where the closure operator $\mbox{cl}_{\qdom,\cdom}$ is defined as in Def. \ref{dfn:interpretation} and the auxiliary interpretation pre-transformer pre$\STp$ acts as follows:
  $$\begin{array}{l}
    pre\STp(\I) \eqdef \{ (f(\ntup{t}{n})\theta \to t) \sharp d \Leftarrow \Pi \mid \text{ there are}\\
    \qquad \text{some } (f(\ntup{t}{n}) \qto{\alpha} r \Leftarrow \ntup{\delta}{m}) \in \Prog,\\
    \qquad \text{some substitution } \theta, \\
    \qquad \text{some set } \Pi \text{ of primitive atomic } \cdom\text{-constraints }, \\
    \qquad \text{some c-term } t \in \Term, \text{ and} \\
    \qquad \text{some qualification values } d_0, d_1, \ldots, d_m \in \aqdom \text{ such that} \\
    \qquad \text{-- } \I \iqcrwldc \delta_i\theta\sharp d_i \Leftarrow \Pi ~ (1 \le i \le m),\\
    \qquad \text{-- } \I \iqcrwldc (r\theta \to t)\sharp d_0 \Leftarrow \Pi, \text{ and}  \\
    \qquad \text{-- } d \dleq \alpha \circ d_i ~ (0 \le i \le m) \\
    \}.
  \end{array}$$
\end{definition}

Proposition \ref{prop:preSTp-entailment} below shows that $pre\STp(\I)$ is closed under $(\qdom, \cdom)$-entailment. Its proof relies on the next technical, but easy result:

\begin{lemma}[Auxiliary Result]
\label{lem:linear-subst}
Given terms $t, t' \in \Term$ and a substitution $\eta$ such that $t$ is linear and $t\eta \sqsubseteq t'$, there is some substitution $\eta'$ such that:
\begin{enumerate}
  \item $t\eta' = t'$ ,
  \item $\eta \sqsubseteq \eta'$ (i.e. $X\eta \sqsubseteq X\eta'$ for all $X \in \Var)$ , and
  \item $\eta = \eta' ~ [\setminus \varset{t}]$ .
\end{enumerate}
\end{lemma}
\begin{proof}
Since $t$ is linear, for each variable $X$ occurring in $t$ there is one single position $p$ such that $X$ occurs in $t$ at position $p$. Let $p_X$ be this position. Since $t\theta \sqsubseteq t'$, there must be a subterm $t'_X$ occurring in $t'$ at position $p_X$ such that $X\eta \sqsubseteq t'_X$. Let $\eta'$ be a substitution such that $X\eta' = t'_X$ for each variable $X$ occurring in $t$, and $Y\eta' = Y\theta$ for each variable $Y$ not occurring in $t$. It is easy to check that $\eta'$ has all the desired properties. \qed
\end{proof}

\begin{proposition}[$pre\STp(\I)$ is closed under $(\qdom, \cdom)$-entailment]
\label{prop:preSTp-entailment}
Assume two qc-facts $\varphi$ and $\varphi'$. If $\varphi \in pre\STp(\I)$ and $\varphi \entail{\qdom,\cdom}$ $\varphi'$, then $\varphi' \in pre\STp(\I)$.
\end{proposition}
\begin{proof}Since $\varphi \in pre\STp(\I)$, there are some
$R_l : (f(\ntup{t}{n}) \qto{\alpha} r \Leftarrow \ntup{\delta}{m}) \in  \Prog$
and some substitution $\theta$ such that
$\varphi : (f(\ntup{t}{n})\theta \to t)\sharp d \Leftarrow \Pi$ and
\begin{itemize}
  \item (1) $\I \iqcrwldc \delta_i\theta \sharp d_i \Leftarrow \Pi ~ (1 \le i \le m)$ ,
  \item (2) $\I \iqcrwldc (r\theta \to t)\sharp d_0 \Leftarrow \Pi$ , and
  \item (3) $d \dleq \alpha \circ d_i ~ (0 \le i \le m)$ .
\end{itemize}

Since $\varphi \entail{\qdom,\cdom} \varphi'$,
we can assume $\varphi' : (f(\ntup{t'}{n}) \to t')\sharp d' \Leftarrow \Pi'$
and a substitution $\sigma$ such that
$t_i\theta\sigma \sqsubseteq t'_i ~ (1 \le i \le n)$,
$t\sigma \sqsupseteq t'$,
(4) $d \dgeq d'$
and $\Pi' \model{\cdom} \Pi\sigma$.

Given that $\ntup{t}{n}$ is a linear tuple of terms, and applying Lemma \ref{lem:linear-subst} with $\eta = \theta\sigma$, we obtain a substitution $\eta'$ satisfying $t_i\eta' = t'_i ~ (1 \le i \le n)$, $\theta\sigma \sqsubseteq \eta'$ and $\theta\sigma = \eta' ~ [\setminus \varset{\ntup{t}{n}}]$.
Now, in order to prove $\varphi' \in pre\STp(\I)$ it suffices to consider $R_l$, $\eta'$ and some some $d'_0$, $d'_1$, \ldots, $d'_m \in \aqdom$ satisfying:
\begin{itemize}
  \item (1') $\I \iqcrwldc \delta_i\eta' \sharp d'_i \Leftarrow \Pi' ~ (1 \le i \le m)$ ,
  \item (2') $\I \iqcrwldc (r\eta' \to t')\sharp d'_0 \Leftarrow \Pi'$ , and
  \item (3') $d' \dleq \alpha \circ d'_i ~ (0 \le i \le m)$ .
\end{itemize}
Let us see that (1'), (2') and (3') hold when choosing $d'_i = d_i ~ (0 \le i \le m)$:

\smallskip
\noindent [1'] For any $1 \le i \le m$ we have $\delta_i\theta\sharp d_i \Leftarrow \Pi \entail{\qdom,\cdom} \delta_i\eta' \sharp d_i \Leftarrow \Pi'$ using $\sigma$, because  $\delta_i\theta\sigma \sqsubseteq \delta_i\eta'$, $d_i \dgeq d_i$ and $\Pi' \model{\cdom} \Pi\sigma$. Therefore (1) $\Rightarrow$ (1') by the entailment property (Lemma \ref{lem:iqcrwl-properties}(\ref{lem:iqcrwl-properties:3})).

\smallskip
\noindent [2']  Similarly as for (1'), $(r\theta \to t) \sharp d_0 \Leftarrow \Pi \entail{\qdom,\cdom} (r\theta' \to t')\sharp d_0\mu \Leftarrow \Pi'$ using $\sigma$, because $r\theta\sigma \sqsubseteq r\eta'$, $t\sigma \sqsupseteq t'$, $d_0 \dgeq d_0$ and $\Pi' \model{\cdom} \Pi\sigma$. Therefore (2) $\Rightarrow$ (2') again by the entailment property (Lemma \ref{lem:iqcrwl-properties}(\ref{lem:iqcrwl-properties:3})).

\smallskip
\noindent [3']  From (3) and (4) we trivially get $d' \dleq \alpha \circ d_i ~ (0 \le i \le m)$. Therefore,  (3') holds when choosing $d'_i = d_i ~ (0 \le i \le m)$. \qed
\end{proof}

As a consequence of the previous proposition, we can establish a stronger relation between $\STp(\I)$ and $pre\STp(\I)$ for non-trivial qc-facts, as given in the following lemma.

\begin{lemma}[$\STp(\I)$ versus $pre\STp(\I)$]
\label{lem:non-trivial-STp}
For any non-trivial qc-fact $\varphi$ one has: $\varphi \in \STp(\I) \Longrightarrow \varphi \in pre\STp(\I)$.
\end{lemma}
\begin{proof}
From $\varphi \in \STp(\I)$ it follows by definition of $\STp$ that $\varphi \in \closure{\qdom,\cdom}{pre\STp(\I)}$. As we are assuming that $\varphi$ is not trivial, there must be some $\psi \in pre\STp(\I)$ such that $\psi \entail{\qdom,\cdom} \varphi$. Then $\varphi \in pre\STp(\I)$ follows from Proposition \ref{prop:preSTp-entailment}. \qed
\end{proof}

The main properties of the interpretation transformer $\STp$ are given in the following proposition.

\begin{proposition}[Properties of interpretation transformers]
Let $\Prog$ be a $\qcflp{\qdom}{\cdom}$-program. Then:
\label{prop:properties-STp}
\begin{enumerate}
  \item $\STp$ is monotonic and continuous.
  \item For any $\I \in \Intdc$: $\I \model{\qdom,\cdom} \Prog \Longleftrightarrow \STp(\I) \subseteq \I$ .
\end{enumerate}
\end{proposition}
\begin{proof}
Monotonicity and continuity are well-known results for similar semantics; see e.g. Prop. 3 in \cite{LRV07}.
Item 2 can be proved as follows:
as an easy consequence of Def. \ref{dfn:model},  $\I \model{\qdom,\cdom} \Prog \iff pre\STp(\I) \subseteq \I$.
Moreover, $pre\STp(\I) \subseteq \I \iff
\closure{\qdom,\cdom}{pre\STp(\I)} \subseteq \closure{\qdom,\cdom}{\I} \iff
\STp(\I) \subseteq \I$,
where the first equivalence is obvious and the second equivalence is due to the equalities
$\closure{\qdom,\cdom}{pre\STp(\I)} = \STp(\I)$ and $\closure{\qdom,\cdom}{\I} = \I$.
Therefore, $\I \model{\qdom,\cdom} \Prog \iff \STp(\I) \subseteq \I$, as desired. \qed
\end{proof}

Finally, we can conclude that the least fixpoint of $\STp$ characterizes the least model of any given $\qcflp{\qdom}{\cdom}$-program $\Prog$, as stated in the following theorem.

\begin{theorem}\label{thm:least-model-lfp}
For every $\qcflp{\qdom}{\cdom}$-program $\Prog$ there exists the \emph{least model} $\Sp = l\!f\!p(\STp) = \union_{k \in \NAT} \STp \!\!\uparrow^k\! (\ibot)$.
\end{theorem}
\begin{proof}Due to a well-known theorem by Knaster and Tarski \cite{Tar55}, a monotonic mapping from a complete lattice into itself always has a least fixpoint which is also its least pre-fixpoint. In the case that the mapping is continuous, its least fixpoint can be characterized as the lub of the sequence of lattice elements obtained by reiterated application of the mapping to the bottom element. Combining these results with Prop. \ref{prop:properties-STp} trivially proves the theorem. \qed
\end{proof}

\subsection*{\it A qualified constraint rewriting logic.}


In order to obtain a logical  view of  program semantics and an alternative characterization of least program models, we define the {\em Qualified Constrained Rewriting Logic for Programs}  $\QCRWLdc$ as the formal system consisting of the six inference rules displayed in Fig. \ref{fig:qcrwl}. Note that  $\QCRWLdc$ is very similar Qualified Constrained Rewriting Logic for Interpretations $\IQCRWLdc$ (see Fig.  \ref{fig:iqcrwl}),
except that the inference rule \textbf{QDF}$_\I$ from $\IQCRWLdc$ is replaced by the inference rule
\textbf{QDF}$_\Prog$ in  $\QCRWLdc$. The inference rules in $\QCRWLdc$ formalize provability of qc-statements from a given program $\Prog$ according to their intuitive meanings.
In particular,  {\bf QDF}$_\Prog$ formalizes the behavior of program rules and attenuation factors that was informally explained in the Introduction, using the set $[\Prog]_\bot$ of {\em program rule instances}.

In the sequel we use the notation $\Prog \qcrwldc \varphi$ to indicate that $\varphi$ can be inferred from $\Prog$ in $\QCRWLdc$. By convention, we agree that no other inference rule is used whenever \textbf{QTI} is applicable. Therefore, trivial qc-statements can only be inferred by rule \textbf{QTI}. As usual in formal inference systems, $\QCRWLdc$ proofs can be represented as trees whose nodes correspond to inference steps. For example, if $\Prog$ is the library program, $\Pi$ is empty, and $\psi$ is

\begin{center}\tt
(guessGenre(book(4,"Beim Hauten der Zwiebel","Gunter Grass", "German","Biography", medium, 432)) --> "Essay")\#0.7
\end{center}

\noindent
then $\Prog \qcrwl{\U}{\rdom} \psi \Leftarrow \Pi$ with a proof tree whose root inference may be chosen as {\bf QDF}$_\Prog$ using
a suitable instance of the fourth program rule for \texttt{guessGenre}.

\begin{figure}[ht]
  \small
  \centering
    \begin{tabular}{|@{\hspace*{0.5cm}}l@{\hspace*{0.5cm}}|}\hline
      \\
      \textbf{QTI} $\quad$ $\displaystyle\frac{}{\quad\varphi\quad}$ ~
      if $\varphi$ is a trivial qc-statement.\\
      \\
      \textbf{QRR} $\quad$ $\displaystyle\frac{}{(v \to v)\sharp d \Leftarrow \Pi}$ ~
      if $v \in \Var \cup B_\cdom$ and $d \in \aqdom$.\\
      \\
      \textbf{QDC} $\quad$ $\displaystyle\frac
        {(~ (e_i \to t_i)\sharp d_i \Leftarrow \Pi ~)_{i = 1 \ldots n}}
        {(c(\ntup{e}{n}) \to c(\ntup{t}{n}))\sharp d \Leftarrow \Pi}$ ~
      \begin{tabular}{l}
        if $c \in DC^n$ and $d \in \aqdom$ \\
        verifies $d \dleq d_i$ $(1 \le i \le n)$.\\
      \end{tabular} \\
      \\
      \textbf{QDF$_{\Prog}$} $\quad$ $\displaystyle\frac
        {(~ (e_i \to t_i)\sharp d_i \Leftarrow \Pi ~)_{i = 1 \ldots n}
          \quad (r \to t)\sharp d'_0  \Leftarrow \Pi
          \quad ( \delta_j\sharp d'_j  \Leftarrow \Pi )_{j=1 \ldots m}}
        {(f(\ntup{e}{n}) \to t)\sharp d \Leftarrow \Pi}$ \\ \\
      if $f \in DF^n$ and $(f(\ntup{t}{n}) \qto{\alpha} r \Leftarrow \delta_1,\ldots,\delta_m) \in [\Prog]_{\bot}$ \\
      where $[\Prog]_\bot = \{ R_l\theta \mid R_l \mbox{ is a rule in } \Prog \mbox{ and } \theta \mbox{ is a substitution}\}$,\\
      and $d \in \aqdom$ verifies $d \dleq d_i ~ (1 \le i \le n)$, $d \dleq \alpha \circ d'_j ~ (0 \le j \le m)$.\\
      \\
      \textbf{QPF} $\quad$ $\displaystyle\frac
        {(~ (e_i \to t_i)\sharp d_i \Leftarrow \Pi ~)_{i=1 \ldots n}}
        {(p(\ntup{e}{n}) \to v)\sharp d \Leftarrow \Pi}$ ~
      if $p \in PF^n$, $v \in \Var \cup DC^0 \cup B_\cdom$, \\ \\
      $\Pi \model{\cdom} p(\ntup{t}{n}) \to v$ and $d \in \aqdom$ verifies $d \dleq d_i ~ (1 \le i \le n)$.\\
      \\
      \textbf{QAC} $\quad$ $\displaystyle\frac
        {(~ (e_i \to t_i)\sharp d_i \Leftarrow \Pi ~)_{i=1 \ldots n}}
        {(p(\ntup{e}{n}) == v)\sharp d \Leftarrow \Pi}$ ~
      if $p \in PF^n$, $v \in \Var \cup DC^0 \cup B_\cdom$,\\ \\
      $\Pi \model{\cdom} p(\ntup{t}{n}) == v$ and $d \in \aqdom$ verifies $d \dleq d_i ~ (1 \le i \le n)$.\\
      \\\hline
    \end{tabular}
  \caption{Qualified Constrained Rewriting Logic for Programs}
  \label{fig:qcrwl}
\end{figure}


The following lemma states the main properties of $\QCRWLdc$.
The proof is similar to that of Lemma \ref{lem:iqcrwl-properties} and omitted here.
The interested reader is also referred to  the proof of Lemma 2 in \cite{LRV07}.

\begin{lemma}[Some properties of $\QCRWLdc$]
\label{lem:qcrwl-properties}
The  three first items of Lemma \ref{lem:iqcrwl-properties} also hold for $\QCRWLdc$, with the natural reformulation of their statements.
More precisely:
\begin{enumerate}
  \item\label{lem:qcrwl-properties:1} Approximation property: For any non-trivial $\varphi$ of the form $(t \to t')\sharp d \Leftarrow \Pi$ where $t, t' \in \Term$, the three following affirmations are equivalent:
  (a) $t \sqsupseteq t'$;
  (b) $\Prog \qcrwldc \varphi$ with an easy proof tree; and
  (c) $\Prog \qcrwldc \varphi$.
  \item\label{lem:qcrwl-properties:2} Primitive c-atoms: For any primitive c-atom $p(\ntup{t}{n}) == v$, one has $\Prog \qcrwldc (p(\ntup{t}{n}) == v)\sharp d \Leftarrow \Pi \iff  \Pi \model{\cdom} p(\ntup{t}{n}) == v$.
  \item\label{lem:qcrwl-properties:3} Entailment property: $\Prog \qcrwldc \varphi$  with a proof tree $\mathcal{T}$
  and $\varphi \entail{\qdom,\cdom} \varphi'$ $\Longrightarrow$ $\Prog \qcrwldc \varphi'$ with a proof tree $\mathcal{T'}$ such that
  $|\mathcal{T'}| \leq |\mathcal{T}|$.
\end{enumerate}
\end{lemma}


The next theorem is the main result in this section. It provides a nice equivalence between $\QCRWLdc$-derivability and semantic consequence in the sense of Definition \ref{dfn:model} ({\em soundness} and {\em completeness} properties), as well as a characterization of least program models in terms of $\QCRWLdc$-derivability ({\em canonicity property}).

\begin{theorem}[$\QCRWLdc$ characterizes program semantics]
\label{thm:correctness}
For any $\qcflp{\qdom}{\cdom}$-program $\Prog$ and any qc-statement $\varphi$, the following three conditions are equivalent:
$$
  (a) \quad \Prog \qcrwldc \varphi \hspace{1cm}
  (b) \quad \Prog \model{\qdom,\cdom} \varphi \hspace{1cm}
  (c) \quad \Sp \iqcrwldc \varphi
$$
Moreover, we also have:
\begin{enumerate}
  \item Soundness: for any qc-statement $\varphi$, $\Prog \qcrwldc \varphi \Longrightarrow \Prog \model{\qdom,\cdom} \varphi$.
  \item Completeness: for any qc-statement $\varphi$, $\Prog \model{\qdom,\cdom} \varphi \Longrightarrow \Prog \qcrwldc \varphi$.
  \item Canonicity: $\Sp = \{\varphi \mid \varphi \text{ is a qc-fact and } \Prog \qcrwldc \varphi\}$.
\end{enumerate}
\end{theorem}
\begin{proof}
Assuming the equivalence between $(a)$, $(b)$ and $(c)$,
soundness and completeness are a trivial consequence of the equivalence between $(a)$ and $(b)$,
and canonicity holds because of the equivalences $\varphi \in \Sp \iff \Sp \iqcrwldc \varphi \iff \Prog \qcrwldc \varphi$,
which follow from the  conservation property from Lemma \ref{lem:iqcrwl-properties} and the equivalence between $(c)$ and $(a)$.
The rest of the proof consists of separate proofs for the three implications
$(a) \Rightarrow (b)$, $(b) \Rightarrow (c)$ and $(c) \Rightarrow (a)$.

\smallskip
\noindent [$(a) \Rightarrow (b)$] We assume $(a)$, i.e.,  $\Prog \qcrwldc \varphi$ with a $\QCRWLdc$-proof tree $\mathcal{T}_\Prog$
including $k \geq 1$ $\QCRWLdc$-inference steps. In order to prove $(b)$ we also assume a qc-interpretation $\I$
such that $\I \model{\qdom,\cdom} \Prog$.
We must prove $\I \iqcrwldc \varphi$ with some $\QCRWLdc$-proof tree $\mathcal{T}_\I$. This follows easily by induction on $k$, using the fact that each $\QCRWLdc$-inference rule \textbf{QRL} is {\em sound} in the following sense: each inference step
$$\displaystyle\frac{~ \varphi_1  ~\cdots~ \varphi_n  ~}{\varphi} \text{ QRL}$$
verifying $\I \iqcrwldc \varphi_i ~ (1 \le i \le n)$ (i.e., the premises are valid in $\I$) also verifies
$\I \iqcrwldc \varphi$ (i.e., the conclusion is valid in $\I$).
For \textbf{QRL} other than \textbf{QDF}$_\Prog$, soundness of \textbf{QRL} does not depend on the assumption $\I \model{\qdom,\cdom} \Prog$;
it can be easily proved by using the homonomous $\IQCRWLdc$-inference rule \textbf{QRL}.
In the case of  \textbf{QDF}$_\Prog$, $\varphi$ has the form $f(\ntup{e}{n}) \to t) \sharp d \Leftarrow \Pi$
and the validity of the premises in $\I$ means the following:
\begin{itemize}
  \item (1) $\I \iqcrwldc (e_i \to t_i) \sharp d_i \Leftarrow \Pi ~ (1 \le i \le n)$ ,
  \item (2) $\I \iqcrwldc (r \to t) \sharp d'_0 \Leftarrow \Pi$ , and
  \item (3) $\I \iqcrwldc \delta_j\sharp d'_j \Leftarrow \Pi ~ (1 \le j \le m)$
\end{itemize}
with $f \in DF^n$,
$(f(\ntup{t}{n}) \qto{\alpha} r \Leftarrow \delta_1, \cdots, \delta_m) \in [\Prog]_\bot$,
$d \dleq d_i ~ (1 \le i \le n)$ and
$d \dleq \alpha \circ d'_j ~ (0 \le j \le m)$.
Then, from the assumption $\I \model{\qdom,\cdom} \Prog$ and Def. \ref{dfn:model} we obtain
\begin{itemize}
  \item (4) $((f(\ntup{t}{n}) \to t)\sharp d \Leftarrow \Pi) \in \I$.
\end{itemize}
Finally, from (1), (4) we conclude that $(f(\ntup{e}{n}) \to t) \sharp d \Leftarrow \Pi$ can de derived by means of a
\textbf{QDF}$_\I$-inference step from premises $(e_i \to t_i) \sharp d_i \Leftarrow \Pi ~ (1 \le i \le n)$.
Therefore, $\I \iqcrwldc (f(\ntup{e}{n}) \to t) \sharp d \Leftarrow \Pi$, as desired.

\smallskip
\noindent [$(b) \Rightarrow (c)$] Straightforward, given that $\Sp \model{\qdom,\cdom} \Prog$,
as proved in Th. \ref{thm:least-model-lfp}.

\smallskip
\noindent [$(c) \Rightarrow (a)$] Let $\varphi$ be any c-statement and assume $\Sp \iqcrwldc \varphi$ with proof tree $\mathcal{T}$.
Note that  $\mathcal{T}$ includes a finite number of \textbf{QDF}$_\I$-inference steps with $\I = \Sp$,
relying on finitely many qc-facts $\psi_i \in \Sp\, (1 \leq i \leq p)$.
As $\Sp = \union_{k\in\NAT} \STp\!\uparrow^k(\ibot)$ because of Th. \ref{thm:least-model-lfp},
there must exist some $k \in \NAT$ such that all the $\psi_i\, (1 \leq i \leq p)$ belong to $\STp\!\uparrow^k(\ibot)$ and thus
$\STp\!\uparrow^k(\ibot) \iqcrwldc \varphi$.
Therefore, it is enough to prove by induction on $k$ that
$$\STp\!\uparrow^k(\ibot) \iqcrwldc \varphi \Longrightarrow \Prog \qcrwldc \varphi \enspace $$

\noindent\emph{Basis} ($k$=0). Assume $\STp\!\uparrow^0(\ibot) \iqcrwldc \varphi$ with $\IQCRWLdc$-proof tree $\mathcal{T}$.
As $\STp\!\uparrow^0(\ibot) = \ibot$, which only includes trivial qc-facts and \textbf{QDF}$_\I$ always uses non-trivial qc-facts,
$\mathcal{T}$  cannot include \textbf{QDF}$_\I$-inference steps. Hence, $\mathcal{T}$  also serves as a $\QCRWLdc$-proof tree
which includes no \textbf{QDF}$_\Prog$-inference steps and proves $\STp\!\uparrow^0(\ibot) \qcrwldc \varphi$. \\

\noindent\emph{Inductive step} ($k{>}$0). Assume $\STp\!\uparrow^{k+1}(\ibot) \iqcrwldc \varphi$ with $\IQCRWLdc$-proof tree $\mathcal{T}$.
Then $\Prog \qcrwldc \varphi$ can be proved by an auxiliary  induction on the size of $\mathcal{T}$, measured as its number of nodes.
The reasoning must distinguish six cases,
according to the $\IQCRWLdc$-inference rule \textbf{QRL} used to infer $\varphi$ at the root of $\mathcal{T}$.
Here we present only the most interesting case, when \textbf{QRL} is \textbf{QDF}$_\I$.
In this case, $\varphi$ is a non-trivial qc-statement of the form $(f(\ntup{e}{n}) \to t)\sharp d \Leftarrow \Pi$,
and $\mathcal{T}$ has the form
$$\displaystyle\frac{(~ (e_i \to t_i)\sharp d_i \Leftarrow \Pi ~)_{i =1\ldots n}}{\varphi : (f(\ntup{e}{n}) \to t)\sharp d \Leftarrow \Pi} \text{ QDF}_\I$$
with non-trivial $\psi : ((f(\ntup{t}{n}) \to t)\sharp d_0 \Leftarrow \Pi) \in \STp\!\uparrow^{k+1}(\ibot)$,
$d \dleq d_i ~ (0 \le i \le n)$,
and $\STp\!\uparrow^{k+1}(\ibot) \iqcrwldc (e_i \to t_i)\sharp d_i \Leftarrow \Pi$
proved by $\IQCRWLdc$-proof trees $\mathcal{T}_i$ wit sizes smaller than the size of $\mathcal{T}$ $(1 \leq i \leq n)$.
Therefore, the inductive hypothesis of the nested induction guarantees
\begin{itemize}
\item (1) $\Prog \qcrwldc (e_i \to t_i)\sharp d_i \Leftarrow \Pi$
with $\QCRWLdc$-proof trees $\mathcal{\hat{T}}_i\  (1 \le i \le n)$
\end{itemize}
On the other hand, Lemma \ref{lem:non-trivial-STp} ensures  $\psi \in pre\STp(\STp\!\uparrow^k(\ibot))$.
Therefore, recalling Def. \ref{dfn:transformers}, there must exist
$f(\ntup{s}{n}) \qto{\alpha} r \Leftarrow \ntup{\delta}{m} \in \Prog$,
a substitution $\theta$ and qualification values $d'_0, d'_1, \ldots, d'_m$
satisfying $s_i\theta = t_i ~ (1 \le i \le n)$ and
\begin{itemize}
  \item (2) $\STp\!\uparrow^k(\ibot) \iqcrwldc \delta_j\theta \sharp d'_j \Leftarrow \Pi ~ (1 \le j \le m)$
  \item (3) $\STp\!\uparrow^k(\ibot) \iqcrwldc (r\theta \to t) \sharp d'_0 \Leftarrow \Pi$
  \item (4) $d_0 \dleq \alpha \circ d'_j ~ (0 \le j \le m)$
\end{itemize}
By the inductive hypothesis of the main induction, applied to (2) and (3), we get:
\begin{itemize}
  \item (5) $\Prog \qcrwldc \delta_j\theta \sharp d'_j \Leftarrow \Pi$
  with $\QCRWLdc$-proof trees $\mathcal{\hat{T'}}_j\  (1 \le j \le m)$
  \item (6) $\Prog \qcrwldc (r\theta \to t) \sharp d'_0 \Leftarrow \Pi$
  with $\QCRWLdc$-proof tree $\mathcal{\hat{T'}}$
\end{itemize}
From  $d \dleq d_i ~ (0 \le i \le n)$ and (4) we also obtain:
\begin{itemize}
  \item (7) $d \dleq d_i ~ (0 \le i \le n), ~ d \dleq \alpha \circ d'_j ~ (0 \le j \le m)$
\end{itemize}
Finally, we can prove $\Prog \qcrwldc \varphi$ with a $\QCRWLdc$-proof tree $\mathcal{\hat{T}}$ of the form:
$$
\displaystyle\frac
{((e_i \to s_i\theta)\sharp d_i \Leftarrow \Pi)_{i=1\ldots n} ~~
(r\theta \to t)\sharp d'_0 \Leftarrow \Pi ~~
(\delta_j\theta \sharp d'_j \Leftarrow \Pi)_{j=1\ldots m}}
{\varphi : (f(\ntup{e}{n}) \to t)\sharp d \Leftarrow \Pi} \text{ QDF}_\Prog
$$
using the program rule instance $(f(\ntup{s}{n}) \qto{\alpha} r \Leftarrow \ntup{\delta}{m})\theta \in [\Prog]_\bot$,
where (5) and (6) provide proof trees for deriving the premises and (7) ensures the additional conditions required by the
\textbf{QDF}$_\Prog$ inference at the root of $\mathcal{\hat{T}}$. \qed
\end{proof}

\subsection{Goals and their Solutions}

In all declarative programming paradigms, programs are generally used by placing goals and computing answers for them.
In this brief subsection we define the syntax of $\qcflp{\qdom}{\cdom}$-goals and we give a declarative characterization of goal solutions,
based on the $\QCRWLdc$ logic. This will allow formal proofs of correctness for the  goal solving methods presented  in Section \ref{Implementing}.

\begin{definition}[$\qcflp{\qdom}{\cdom}$-Goals and their Solutions]
\label{dfn:goals}
Assume a a countable set $\War$ of so-called  {\em qualification variables} $W$, disjoint from $\Var$ and $\cdom$'s signature $\Sigma$, and a $\qcflp{\qdom}{\cdom}$-program $\Prog$.  Then:
\begin{enumerate}
  \item A {\em goal} $G$ for $\Prog$ has the form $\delta_1\sharp W_1, \ldots, \delta_m\sharp W_m \sep W_1 \dgeq \beta_1, \ldots, W_m \dgeq \beta_m$, abbreviated as $(~ \delta_i\sharp W_i,\ W_i \!\!\dgeq\!\! \beta_i ~)_{i = 1 \ldots m}$, where $\delta_j\sharp W_j ~ (1 \leq j \leq m)$ are atomic $\cdom$-constraints annotated with different qualification variables $W_i$, and $W_i \dgeq \beta_i$ are so-called {\em threshold conditions}, with $\beta_i \in \aqdom ~ (1 \leq i \leq m)$.
  \item A {\em solution} for $G$ is any triple $\langle \sigma, \mu, \Pi \rangle$ such that $\sigma$ is a substitution, $\mu$ is a $\qdom$-valuation, $\Pi$ is a finite set of atomic primitive $\cdom$-constraints, and the following two conditions hold for all $1 \leq i \leq m$: $W_i\mu = d_i \dgeq \beta_i$, and $\Prog \qcrwldc (\delta_i\sigma)\sharp d_i \Leftarrow \Pi$. The set of all solutions for $G$ is noted $Sol_\Prog(G)$.\qed
\end{enumerate}
\end{definition}

Thanks to the \emph{Canonicity} property of Theorem \ref{thm:correctness}, solutions of $\Prog$ are valid in the least model $\Sp$ and hence in all models of $\Prog$. A goal for the library program and one solution for it have been presented in the Introduction. In this particular example, $\Pi = \emptyset$ and the $\QCRWL{\U}{\rdom}$ proof needed to check the solution according to Definition \ref{dfn:goals} can be formalized by following the intuitive ideas sketched in the Introduction.

\section{Implementation by Program Transformation} \label{Implementing}

Goal solving in instances of the $\cflp{\cdom}$ scheme from \cite{LRV07} has been  formalized by
means of {\em constrained narrowing} procedures as e.g. \cite{LRV04,Vad05}, and is supported by
systems such as {\sf Curry} \cite{curry} and $\toy$ \cite{toy}. In this section we present a
semantically correct transformation from $\qcflp{\qdom}{\cdom}$ into the first-order fragment of
$\cflp{\cdom}$ which can be used for implementing goal solving in $\qcflp{\qdom}{\cdom}$.


By abuse of notation, the first-order fragment of the $\cflp{\cdom}$ scheme will be noted simply as
$\cflp{\cdom}$ in the sequel. A formal description of $\cflp{\cdom}$ can be found in \cite{LRV07};
it is easily derived from the previous Section \ref{QCFLP} by simply omitting everything related to
qualification domains and values. Programs $\Prog$ are sets of program rules of the form
$f(\ntup{t}{n}) \to r \Leftarrow \Delta$, with no attenuation factors attached. Program semantics
relies on inference mechanisms for deriving {\em c-staments} from programs. In analogy to Def.
\ref{dfn:qc-statements}, a c-statement $\varphi$ may be a c-production $e \to t \Leftarrow \Pi$ or
a c-atom $\delta \Leftarrow \Pi$. In analogy to Def. \ref{dfn:interpretation}, {\em
c-interpretations} are defined as sets of c-statements closed under a $\cdom$-entailment relation.
Program models and semantic consequence are defined similarly as in Def. \ref{dfn:model}. Results
similar to Th. \ref{thm:least-model-lfp} and Th. \ref{thm:correctness} can be obtained to
characterize program semantics in terms of an interpretation transformer and a rewriting logic
$\CRWLc$, respectively.

For the purposes of this section it is enough to focus on $\CRWLc$, which is a formal system
consisting of the six inference rules displayed in Fig. \ref{fig:crwl}. They are quite similar to
the $\QCRWLdc$-inference rules from Fig. \ref{fig:qcrwl}, except that attenuation factors and
qualification values are absent.

\begin{figure}[h]
  \small
  \centering
    \begin{tabular}{|@{\hspace*{0.5cm}}l@{\hspace*{0.5cm}}|}\hline
      \\
      \textbf{TI} $\quad$ $\displaystyle\frac{}{\quad\varphi\quad}$ ~
      if $\varphi$ is a trivial c-statement.\\
      \\
      \textbf{RR} $\quad$ $\displaystyle\frac{}{~ v \to v \Leftarrow \Pi ~}$ ~
      if $v \in \Var \cup B_\cdom$.\\
      \\
      \textbf{DC} $\quad$ $\displaystyle\frac
        {(~ e_i \to t_i \Leftarrow \Pi ~)_{i = 1 \ldots n}}
        {c(\ntup{e}{n}) \to c(\ntup{t}{n}) \Leftarrow \Pi}$ ~
        if $c \in DC^n$. \\
      \\
      \textbf{DF$_{\Prog}$} $\quad$ $\displaystyle\frac
        {(~ e_i \to t_i \Leftarrow \Pi ~)_{i = 1 \ldots n}
          \quad r \to t  \Leftarrow \Pi
          \quad (~ \delta_j \Leftarrow \Pi ~)_{j=1 \ldots m}}
        {f(\ntup{e}{n}) \to t \Leftarrow \Pi}$ \\ \\
      if $f \in DF^n$ and $(f(\ntup{t}{n}) \qto{\alpha} r \Leftarrow \delta_1,\ldots,\delta_m) \in [\Prog]_{\bot}$ \\
      where $[\Prog]_\bot = \{ R_l\theta \mid R_l \mbox{ is a rule in } \Prog \mbox{ and } \theta \mbox{ is a substitution}\}$.\\
      \\
      \textbf{PF} $\quad$ $\displaystyle\frac
        {(~ e_i \to t_i \Leftarrow \Pi ~)_{i=1 \ldots n}}
        {p(\ntup{e}{n}) \to v \Leftarrow \Pi}$ ~
      \begin{tabular}{l}
        if $p \in PF^n$, $v \in \Var \cup DC^0 \cup B_\cdom$ \\
        and $\Pi \model{\cdom} p(\ntup{t}{n}) \to v$.
      \end{tabular} \\
      \\
      \textbf{AC} $\quad$ $\displaystyle\frac
        {(~ e_i \to t_i \Leftarrow \Pi ~)_{i=1 \ldots n}}
        {p(\ntup{e}{n}) == v \Leftarrow \Pi}$ ~
      \begin{tabular}{l}
        if $p \in PF^n$, $v \in \Var \cup DC^0 \cup B_\cdom$ \\
        and $\Pi \model{\cdom} p(\ntup{t}{n}) == v$.
      \end{tabular} \\
      \\\hline
    \end{tabular}
  \caption{First Order Constrained Rewriting Logic}
  \label{fig:crwl}
\end{figure}

The notation $\Prog \crwlc \varphi$  indicates that $\varphi$ can be inferred from $\Prog$ in
$\CRWLc$. In analogy to the Canonicity Property from Th.  \ref{thm:correctness}, it is possible to
prove that the least model of $\Prog$ w.r.t. set inclusion can be characterized as $\Sp = \{\varphi
\mid \varphi \text{ is a c-fact and } \Prog \crwlc \varphi\}$. Therefore, working with formal
inference in the rewrite logics $\QCRWLdc$ and $\CRWLc$ is sufficient for proving the semantic
correctness of the transformations presented in the rest of this section.


The following definition is similar to Def. \ref{dfn:goals}. It will be useful for proving the
correctness  of the goal solving procedure for $\qcflp{\qdom}{\cdom}$-goals discussed in the final
part of this section.

\begin{definition}[$\cflp{\cdom}$-Goals and their Solutions]
\label{dfn:cflpgoals} Assume a $\cflp{\cdom}$-program $\Prog$. Then:
\begin{enumerate}
  \item A {\em goal} $G$ for $\Prog$ has the form $\delta_1, \ldots, \delta_m$ where $\delta_j$ are atomic $\cdom$-constraints.
  \item A {\em solution} for $G$ is any pair $\langle \sigma, \Pi\rangle$ such that $\sigma$ is a substitution, $\Pi$ is a finite set of atomic primitive
  $\cdom$-constraints, and $\Prog \crwlc \delta_j\sigma \Leftarrow \Pi$ holds for $1 \leq j \leq m$.
  The set of all solutions for $G$ is noted $Sol_\Prog(G)$.\qed
\end{enumerate}
\end{definition}



Now we are ready to describe a semantically correct transformation from $\qcflp{\qdom}{\cdom}$ into
$\cflp{\cdom}$. The transformation goes from a source signature $\Sigma$ into a target signature
$\Sigma'$ such that each $f \in DF^n$ in $\Sigma$ becomes $f' \in DF^{n+1}$ in $\Sigma'$, and all the other
symbols in $\Sigma$ remain the same in $\Sigma'$. There are four group of transformation rules
displayed in Figure \ref{fig:transformation} and designed to transform expressions, qc-statements,
program rules and goals, respectively. The transformation works by introducing fresh  qualification
variables $W$ to represent the qualification values attached to the results of calls to defined
functions, as well as qualification constraints to be imposed on the values of qualification
variables. Let us comment the four groups of rules in order.

\begin{figure}
  \small
  \centering
    \begin{tabular}{|@{\hspace*{0.5cm}}l@{\hspace*{0.5cm}}|}\hline
      \\
      \bf Transforming Expressions \\
      \\
      \textbf{TAE} $\quad$ $\displaystyle\frac
      {}
      {~ \transform{v} = (v,\emptyset, \emptyset) ~}$
      ~ if  $v \in \Var \cup B_\cdom$. \\
      \\
      \textbf{TCE$_1$} $\quad$ $\displaystyle\frac
      {(~ \transform{e_i} = (e_i',\Omega_i,\W_i) ~)_{i=1\ldots n}}
      {~ \transform{h(\ntup{e}{n}) } = (h(\ntup{e'}{n}),\, \bigcup_{i=1}^n \Omega_i,\, \bigcup_{i=1}^n \W_i) ~}$
      ~ if $h \in DC^n \cup PF^n$. \\
      \\
      \textbf{TCE$_2$} $\quad$ $\displaystyle\frac
      {(~ \transform{e_i} = (e_i',\Omega_i,\W_i) ~)_{i=1\ldots n}}
      {~ \transform{f(\ntup{e}{n}) } = (f'(\ntup{e'}{n},W), \Omega', \{W\}) ~}$ \\
      \\
      if $f \in DF^n$ and $W$ is a fresh variable,\\
      where $\Omega' = (\bigcup_{i=1}^n \Omega_i) \cup \{ \qval{W} \} \cup \{ \encode{W \dleq W'} \mid W' \in \bigcup_{i=1}^n \W_i \}$. \\
      \\\hline
 \\
      \bf Transforming qc-Statements \\
      \\
      \textbf{TP} $\quad$ $\displaystyle\frac
      {\transform{e} = (e', \Omega,\, \W)}
      {~ \transform{(e \rightarrow t)} = (e'\rightarrow t,\, \Omega ,\W) ~}$ \\
      \\
      \textbf{TA} $\quad$ $\displaystyle\frac
      {(~ \transform{e_i} = (e_i', \Omega_i, \W_i)  ~)_{i=1 \dots n}}
      {~ \transform{(p(\ntup{e}{n}) == v)} = (~ p(\ntup{e'}{n})==v, \bigcup_{i=1}^n \Omega_i , \bigcup_{i=1}^n  \W_i  ~) ~}$ \\
      \\
      if $p \in PF^n$, $v \in \Var \cup DC^0 \cup B_\cdom$. \\
       \\
      \textbf{TCS} $\quad$ $\displaystyle\frac
      {\transform{\psi} = (\psi', \Omega, \W)}
      {~ \transform{(\psi \sharp d \Leftarrow \Pi)} = (\psi' \Leftarrow \Pi, \Omega \cup \{ \encode{d \dleq W} \mid W \in \W\})) ~}$ \\
      \\
     if $\psi$ is of the form $e \to t$ or $p(\ntup{e}{n}) == v$ and $d \in D_\qdom$. \\
      \\\hline
      \\
      \bf Transforming Program Rules \\
      \\
      \textbf{TPR} $\quad$ $\displaystyle\frac
      {\transform{r} = (r',\Omega_r,\W_r) \qquad (~ \transform{\delta_i} = (\delta_i',\Omega_i, \W_i) ~)_{i=1 \dots m}}
      {~
        \begin{array}{l}
          \transform{(f(\ntup{t}{n}) \qto{\alpha} r \Leftarrow \delta_1, \dots, \delta_m)} = \\
          $\qquad$
          \begin{array}{ll}
            f'(\ntup{t}{n},W) \to r' \Leftarrow & \qval{W}, ~ \Omega_r, ~ (\encode{W \dleq \alpha \circ W'})_{W'\in \W_r}, \\
            & (~ \Omega_i, ~ (\encode{W \dleq \alpha \circ W'})_{W' \in \W_i}, ~ \delta'_i ~)_{i=1 \ldots m} \\
          \end{array}
        \end{array}
      ~}$ \\
      where $W$ is a fresh variable. \\
      \\\hline
      \\
      \bf Transforming Goals \\
      \\
      \textbf{TG} $\quad$ $\displaystyle\frac
      {(~ \transform{\delta_i} = (\delta_i',\Omega'_i,\W'_i) ~)_{i=1\dots m}}
      {~
      \begin{array}{l}
        \transform{((~ \delta_i \sharp W_i, W_i \dgeq \beta_i ~)_{i=1\dots m})} = \\
        $\qquad$ (~ \Omega_i', ~ \qval{W_i}, ~ (\encode{W_i \dleq W'})_{W' \in \W'_i}, ~ \encode{W_i \dgeq \beta_i}, ~ \delta'_i ~)_{i=1\dots m}
      \end{array}
      ~}$ \\
      \\\hline
    \end{tabular}
  \caption{Transformation rules}
  \label{fig:transformation}
\end{figure}

Transforming any expression $e$ yields a triple $\transform{e} = (e', \Omega, \W)$, where $\Omega$
is a set of qualification constraints and $\W$ is the set of qualification variables occurring in
$e'$ at outermost positions. This set is relevant because the qualification value attached to $e$
cannot exceed the infimum in $\qdom$ of the values of the variables $W \in \W$, and $\transform{e}$
is computed by recursion on $e$'s syntactic structure as specified by the transformation rules {\bf
TAE}, {\bf TCE$_1$} and {\bf TCE$_2$}. Note that {\bf TCE$_2$} introduces a new qualification
variable $W$ for each call to a defined function $f \in DF^n$ and builds a set $\Omega'$ of
qualification constraints ensuring that $W$ must be interpreted as a qualification value not
greater than the qualification values attached to $f$'s arguments. {\bf TCE$_1$} deals with calls
to constructors and primitive functions just by collecting information from the arguments, and
{\bf TAE} is self-explanatory.

Unconditional productions and atomic constraints are transformed by means of {\bf TP} and {\bf TA},
respectively, relying on the transformation of expressions in the obvious way. Relying on {\bf TP}
and {\bf TA},  {\bf TCS} transforms any qc-statement of the form $\psi \sharp d \Leftarrow \Pi$
into a c-statement  whose conditional part includes, in addition to $\Pi$, the qualification
constraints $\Omega$ coming from $\transform{\psi}$ and extra qualification constraints ensuring
that $d$ is not greater than allowed by $\psi$'s qualification.

Program rules are transformed by {\bf TPR}. Transforming the left-hand side $f(\ntup{t}{n})$
introduces a fresh symbol $f' \in DF^{n+1}$ and a fresh qualification variable $W$. The transformed
right-hand side $r'$ comes from $\transform{r}$, and the transformed conditions are obtained from
the constraints coming from $\transform{r}$ and $\transform{\delta_i} (1 \leq i \leq m)$ by adding
extra qualification constraints to be imposed on $W$, namely $\qval{W}$ and $(\encode{W \dleq
\alpha \circ W'})_{W' \in \W'}$, for $\W' = \W_r$ and $\W' = \W_i$ $(1 \leq i \leq m)$. By
convention, $(\encode{W \dleq \alpha \circ W'})_{W' \in \W'}$ is understood as $\encode{W \dleq
\alpha}$ in case that $\W' = \emptyset$. The idea is that $W$'s value cannot exceed the infimum in
$\qdom$ of all the values $\alpha \circ \beta$, for the different $\beta$ coming from the
qualifications  of $r$ and $\delta_i$ $(1 \leq i \leq m)$.

Finally, {\bf TG} transforms a goal $(~ \delta_i \sharp W_i,\, W_i \dgeq \beta_i ~)_{i=1\dots m}$
by transforming each atomic constraint $\delta_i$ and adding $\qval{W_i}$, ($\encode{W_i \dleq
W'})_{W' \in \W'_i}$ and $\encode{W_i \dgeq \beta_i}$ $(1 \leq i \leq m)$ to ensure that each $W_i$
is interpreted as a qualification value not bigger than the qualification computed for $\delta_i$
and satisfying the threshold condition $W_i \dgeq \beta_i$. In case that $\W'_i = \emptyset$,
$(\encode{W_i \dleq W'})_{W' \in \W'_i}$ is understood as $\encode{W_i \dleq \tp}$.


The result of applying {\bf TPR} to all the program rules of a program $\Prog$ will be noted as
$\transform{\Prog}$. The following theorem proves that $\QCRWLdc$-derivability from $\Prog$
corresponds to $\CRWLc$-derivability from $\transform{\Prog}$. Since program semantics in
$\qcflp{\qdom}{\cdom}$ and in $\cflp{\cdom}$ is characterized by, respectively, derivability in
$\QCRWLdc$ and in $\CRWLc$, the program transformation is semantically correct. The theorem uses an
auxiliary lemma we are proving first which indicates that the constraints obtained when
transforming a qc-statement always admit a solution.

\begin{lemma} \label{lemma:trivialaux}
Let $\varphi = \psi\sharp d \Leftarrow \Pi$ be a qc-statement such that $\transform{\varphi} =
(\psi' \Leftarrow \Pi,\Omega')$. Then exists $\rho : \varset{\Omega'} \to \aqdom$ solution of
$\Omega'$.
\end{lemma}
\begin{proof}
$\transform{\varphi}$ is obtained by the transformation rule {\bf TCS} of Figure
\ref{fig:transformation}. This rule needs to obtain $\transform{\psi}$ which can be done using
either the transformation rule {\bf TP} or {\bf TA} of the same figure. In the case of using
\textbf{TP}, $\psi$ must be of the form $(e \rightarrow t)$ and $\Omega'$ will be of the form
$\Omega \cup \{ \encode{d \dleq W} \mid W \in \W\}$, with $\Omega, \W$ such that $\transform{e} = (e',
\Omega, \W)$. Checking the transformation rules for expressions (again Figure
\ref{fig:transformation}) we see that $\Omega$  is a set of constraints where each element is
either of the form $\encode{W \dleq W'}$ or $\qval{W}$, with $W, W' \in \War$. Then $\rho$ can be
defined assigning $\tp$ to every variable $W$ occurring in either $\Omega'$ or $\W$. The case
corresponding to the transformation rule {\bf TA} is analogous. \qed
\end{proof}

\begin{theorem}
\label{theo:transformation} Let $\Prog$ be a $\qcflp{\qdom}{\cdom}$-program and $\psi\sharp d
\Leftarrow \Pi$ a qc-statement such that $\transform{(\psi\sharp d \Leftarrow \Pi)} = (\psi'
\Leftarrow \Pi,\Omega')$. Then the two following statements are equivalent:
\begin{enumerate}
  \item \label{theo:transformation:1} $\Prog \qcrwl{\qdom}{\cdom} \psi \sharp d \Leftarrow \Pi$.
  \item \label{theo:transformation:2} $\transform{\Prog} \crwl{\cdom} \psi'\!\rho \Leftarrow \Pi$ for some $\rho \in \Sol{\cdom}{\Omega'}$ such that $\domset{\rho} = \varset{\Omega'}$.
\end{enumerate}
\end{theorem}
\begin{proof} We prove the equivalence separately proving each implication.

\smallskip
\noindent $[\mathit{\ref{theo:transformation:1}}. \Rightarrow
\mathit{\ref{theo:transformation:2}}.]$ (\emph{Transformation completeness}). Assume $\Prog
\qcrwl{\qdom}{\cdom} \psi \sharp d \Leftarrow \Pi$ by means of a $\QCRWLdc$ proof tree $T$ with $k$
nodes. By induction on $k$ we show the existence of a $\CRWLc$ proof tree $T'$ witnessing
$\transform{\Prog} \crwl{\cdom} \psi'\!\rho \Leftarrow \Pi$ for some $\rho \in \Solc{\Omega'}$ such
that $\domset{\rho} = \varset{\Omega'}$.

\smallskip
\noindent \emph{Basis} ($k$=1). If $T$ contains only one node the $\QCRWLdc$ inference step applied
at the root must be one of the following:
\begin{itemize}
  \item {\bf QTI}. In this case $\psi\sharp d \Leftarrow \Pi$ is a trivial qc-statement, and we take $\rho$ as the substitution defined in Lemma \ref{lemma:trivialaux}. By Def. \ref{dfn:qc-statements}, $\psi\sharp d \Leftarrow \Pi$ trivial implies either $\psi = e \to \bot$ or $\Unsat{\cdom}{\Pi}$. In the first case $\psi' = e' \to \bot$ and therefore $\psi'\!\rho \Leftarrow \Pi$ is trivial. Analogously, if $\Unsat{\cdom}{\Pi}$ then $\psi'\!\rho \Leftarrow \Pi$ is trivial as well. Hence $T'$ consists of a single node $\psi'\!\rho \Leftarrow \Pi$ with a {\bf TI} inference step at its root.

  \item {\bf QRR}. In this case $\psi = t \to t$ for some $t \in \Var \cup \B_\cdom$, and $\transform{(\psi\sharp d \Leftarrow \Pi)} = (t \to t \Leftarrow \Pi, \emptyset)$
(applying the transformation rules {\bf TCS}, {\bf TP} and {\bf TAE}  to obtain $\transform{t} =
(t,\emptyset,\emptyset)$). Therefore $\rho$ can be defined as the identity substitution and prove
$\transform{\Prog} \crwl{\cdom} \psi'\!\rho \Leftarrow \Pi$ by using a single {\bf RR} inference
step.

  \item {\bf QDC}. In this case $\psi = c \to c$ and $\transform{(\psi\sharp d \Leftarrow \Pi)}= (c \to c \Leftarrow \Pi, \emptyset)$ (applying the transformation rules {\bf TCS}, {\bf TP} and {\bf TCE$_1$} for $\transform{c} = (c,\emptyset,\emptyset)$). Therefore $\rho$ can be defined as the identity substitution and prove $\transform{\Prog} \crwl{\cdom} \psi'\!\rho \Leftarrow \Pi$ by using a single {\bf DC} inference step.
\end{itemize}

\smallskip
\noindent \emph{Inductive step} ($k{>}$1). The $\QCRWLdc$ inference step applied at the root must
be one of the following:

\begin{itemize}
  \item {\bf QDC}. In this case $\psi = c(\ntup{e}{n}) \to c(\ntup{t}{n})$ and the first inference step is of the form
$$
  \displaystyle\frac
    {(~ (e_i \to t_i)\sharp d_i \Leftarrow \Pi ~)_{i = 1 \ldots n}}
    {(c(\ntup{e}{n}) \to c(\ntup{t}{n}))\sharp d \Leftarrow \Pi}
$$
with $d \dleq d_i$ $(1 \le i \le n)$. In order to obtain $\transform{\psi\sharp d \Leftarrow \Pi}$
we apply the transformation rules as follows:
\begin{itemize}
  \item By the transformation rule {\bf TCE$_1$}, $$\transform{c(\ntup{e}{n}) } = (c(\ntup{e'}{n}),\, \bigcup_{i=1}^n \Omega_i,\, \bigcup_{i=1}^n \W_i)$$ with $\transform{e_i} = (e_i',\Omega_i,\W_i)$ for $i=1 \dots n$.
  \item By {\bf TP} and with the result of the previous step, $$\transform{\psi} = \transform{(c(\ntup{e}{n}) \to c(\ntup{t}{n}))} = (c(\ntup{e'}{n}) \to c(\ntup{t}{n}),\, \bigcup_{i=1}^n \Omega_i,\, \bigcup_{i=1}^n \W_i) \enspace .$$
  \item And finally from $\transform{\psi}$ and by {\bf TCS}, $$\transform{(\psi\sharp d \Leftarrow \Pi)} = (c(\ntup{e'}{n}) \to c(\ntup{t}{n}) \Leftarrow \Pi, \Omega') \enspace ,$$ with $$ \Omega' =  \bigcup_{i=1}^n \Omega_i \cup \{\encode{d \dleq W} \mid W \in \bigcup_{i=1}^n \W_i\} \enspace .$$
\end{itemize}

From the premises ${(~ (e_i \to t_i)\sharp d_i \Leftarrow \Pi ~)_{i = 1 \ldots n}}$ of the {\bf
QDC} step and by the induction hypothesis we have that $\transform{\Prog} \crwl{\cdom} (e'_i \to
t_i)\rho_i \Leftarrow \Pi$, $i=1 \ldots n$ for some substitutions $\rho_i : \varset{\Omega'_i} \to
\aqdom$ solution of $$\Omega'_i = \Omega_i \cup \{\encode{d_i \dleq W} \mid W \in \W_i\}$$ for $i=1
\dots n$. Since $\varset{\Omega'_i} \cap \varset{\Omega'_j} = \emptyset$ for every $1 \leq i,j \leq
n$, $i \neq j$, and $\varset{\Omega'} = \bigcup_{i=1}^n \varset{\Omega'_i}$, we can define a new
substitution $\rho : \varset{\Omega'} \to \aqdom$ as $\rho = \biguplus_{i=1}^n \rho_i$. It is easy
to check that $\rho$ is solution of $\Omega'$:
\begin{itemize}
  \item It is solution of every $\Omega'_i$ for $i=1 \dots n$, since $\rho {\upharpoonright} \varset{\Omega'_i} = \rho_i$. Therefore it is solution of $\bigcup_{i=1}^n \Omega_i$.
  \item It is a solution of $\{\encode{d \dleq W} \mid W \in \bigcup_{i=1}^n \W_i\}$ because as solution of $\Omega'_i$ for $i=1 \dots n$, $\rho$ is solution of $\{\encode{d_i \dleq W} \mid W \in  \W_i\}$, and by the hypothesis of {\bf QDC} $d \dleq d_i$.
\end{itemize}
Therefore we prove $\transform{\Prog} \crwl{\cdom} (c(\ntup{e'}{n})\rho \to c(\ntup{t}{n}))\rho
\Leftarrow \Pi$ with a proof tree $T'$ which starts with a {\bf DC} inference rule of the form
$$
  \displaystyle\frac
    {((~ e'_i \to t_i)\rho \Leftarrow \Pi ~)_{i = 1\ldots n}}
    {(c(\ntup{e'}{n}) \to c(\ntup{t}{n}))\rho \Leftarrow \Pi} \enspace .
$$

In order to justify that $\transform{\Prog} \crwl{\cdom} (e'_i \to t_i)\rho \Leftarrow \Pi$ for
each $i=1 \dots n$, we observe that the only variables of $e'_i \to t_i$ that can be affected by
$\rho$ are those introduced in $e'_i$ by the transformation, and that therefore $(e'_i \to t_i)\rho
= (e'_i \to t_i)\rho_i$ for $i=1 \dots n$, and these premises correspond to the inductive
hypotheses of this case.

  \item {\bf QDF$_\Prog$}. In this case $\psi = f(\ntup{e}{n}) \to t$ and the inference step applied at the root is of the form
$$
  \displaystyle\frac
    {(~ (e_i \to t_i\theta)\sharp d_i \Leftarrow \Pi ~)_{i = 1 \ldots n}
      \quad (r\theta \to t)\sharp d'_0  \Leftarrow \Pi
      \quad (~ \delta_j\theta\sharp d'_j  \Leftarrow \Pi ~)_{j = 1 \ldots m}}
    {(f(\ntup{e}{n}) \to t)\sharp d \Leftarrow \Pi}
$$
for some program rule $R_l = (f(\ntup{t}{n}) \qto{\alpha} r \Leftarrow \ntup{\delta}{m}) \in \Prog$
and substitution $\theta$ such that $R_l\theta \in [\Prog]_{\bot}$, and with $d \dleq d_i ~ (1 \le
i \le n)$ and $d \dleq \alpha \circ d'_j ~ (0 \le j \le m)$.

The inductive hypotheses in this case are:
\begin{enumerate}
  \item $\transform{\Prog} \crwl{\cdom} (e'_i \to t_i\theta)\rho_i \Leftarrow \Pi$ for $i=1 \dots n$, with $\transform{e_i} = (e'_i,\Omega_i,\W_i)$ and $\rho_i$ solution of $\Omega'_i = \Omega_i \cup \{\encode{d_i \dleq W'} \mid W' \in \W_i\}$, for $i=1 \dots n$.
  \item $\transform{\Prog} \crwl{\cdom} (r'\theta \to t)\rho'_0 \Leftarrow \Pi$, with $\transform{r} = (r',\Omega_r,\W'_0)$ (it is easy to check that if $\transform{r} = (r',\Omega_r,\W'_0)$ then $\transform{(r\theta)} = (r'\theta, \Omega_r, \W'_0)$ for every substitution $\theta$), and $\rho'_0$ solution of $\Omega'_r = \Omega_r \cup \{\encode{d'_0 \dleq W'} \mid W' \in \W'_0\}$.
  \item $\transform{\Prog} \crwl{\cdom} (\delta'_j\theta)\rho'_j \Leftarrow \Pi$ with $\transform{\delta_j} = (\delta_j',\Omega_{\delta_j}, \W'_j)$ for ${j=1 \dots k}$ (it is easy to check that if $\transform{\delta_j} = (\delta'_j,\Omega_{\delta_j}, \W'_j)$ then $\transform{(\delta_j\theta)} = (\delta'_j\theta,\Omega_{\delta_j}, \W'_j)$ for every substitution $\theta$ and $j=1 \dots k$). The substitution $\rho'_j$ is solution of $\Omega'_{\delta_j} = \Omega_{\delta_j} \cup \{\encode{d'_j \dleq W'} \mid W' \in \W'_j\}$ for $j=1 \dots m$.
\end{enumerate}

In this case, $\transform{(\psi\sharp d \Leftarrow \Pi)}$\! is obtained by means of the
transformation rule {\bf TCS}. This rule asks first for the transformation of the qualified
statement $(f(\ntup{e}{n}) \to t)\sharp d$, which can be obtained by rule {\bf TP}, and this one
requires the transformation of $f(\ntup{e}{n})$, provided by rule rule {\bf TCE$_2$}. Let's see it:

$$
  \displaystyle\frac
    {
      \quad \displaystyle\frac
        {
          \quad \displaystyle\frac
            {
              \begin{array}{l}
                (~ \transform{e_i} = (e_i',\Omega_i,\W_i) ~)_{i = 1 \ldots n} \\ \\
              \end{array}
            }
            {
              \begin{array}{l} \\
                \transform{f(\ntup{e}{n})} = (~ f(\ntup{e'}{n}, W), \\
                \qquad (\union_{i=1}^{n} \Omega_i) \cup \{\qval{W}\} ~\cup\\
                \qquad \{\encode{W \dleq W'} \mid W' \in \union_{i=1}^{n} \W_i\}, ~ \{W\} ~) \\ \\
              \end{array}
            } \quad \mbox{\bf TCE$_2$}
        }
        {
          \begin{array}{l}  \\
            \transform{(f(\ntup{e}{n}) \to t)} =  (~ f(\ntup{e'}{n}, W) \to t, \\
            \qquad (\union_{i=1}^{n} \Omega_i) \cup \{\qval{W}\} ~\cup\\
            \qquad \{\encode{W \dleq W'} \mid W' \in \union_{i=1}^{n} \W_i\}, ~\{W\} ~) \\ \\
          \end{array}
        } \quad \mbox{\bf TP}
    }
    {
      \begin{array}{l} \\
        \transform{((f(\ntup{e}{n})\to t)\sharp d \Leftarrow\Pi)} = (~ f(\ntup{e'}{n},W)\to t \Leftarrow \Pi, \\
        \qquad (\union_{i=1}^{n} \Omega_i) \cup \{\qval{W}\} ~\cup\\
        \qquad \{\encode{W \dleq W'} \mid W' \in \union_{i=1}^{n} \W_i\} \cup \{\encode{d \dleq W}\} ~) \\
      \end{array}
    } \quad \mbox{\bf TCS}
$$
Therefore
$$
\Omega' = (\bigcup_{i=1}^n \Omega_i) \cup \{ \qval{W} \} \cup \{\encode{W \dleq W'} \mid W' \in
\bigcup_{i=1}^n \W_i \} \cup \{\encode{d \dleq W}\} \enspace .
$$
We define a new substitution
$$
\rho = \biguplus_{i=1}^n \rho_i \uplus \rho'_0 \uplus \biguplus_{j=1}^m \rho'_j \uplus \{W
\!\mapsto d\} \enspace .
$$

It is straightforward to check that $\rho$ is a solution for $\Omega'$ because $\rho$ is solution
of:
\begin{itemize}
  \item Each $\Omega_i ~ (1 \le i \le n)$, because $\rho_i$ is solution of $\Omega'_i$ which contains $\Omega_i$ (see inductive hypothesis 1) and $\rho$ is an extension of $\rho_i$.
  \item  $\{\qval{W}\}$ because $\qval{W}\rho = \qval{d}$ which holds by definition.
  \item $\{\encode{W \dleq W'} \mid W' \in \union_{i=1}^n \W_i\}$ because $W\!\rho = d$, $\rho$ is solution of $ \{\encode{d_i \dleq W'} \mid W' \in \W_i\}$ for each $i=1 \dots n$ (see inductive hypothesis 1), and $d \dleq d_i ~ (1 \le i \le n)$ by the hypotheses of the inference rule \textbf{QDP$_\Prog$}.
  \item $\{\encode{d \dleq W}\}$ since $W\!\rho = d$ and trivially $d \dleq d$.
\end{itemize}

The transformed of the program rule $R_l = (f(\ntup{t}{n}) \qto{\alpha} r \Leftarrow
\ntup{\delta}{m}) \in \Prog$ will be a program rule in $\transform{\Prog}$ of the form:
$$
  \begin{array}{ll}
  \transform{(R_l)} = (f(\ntup{t}{n},W) \rightarrow r' \Leftarrow & \qval{W}, \Omega_r, (\encode{W \dleq \alpha \circ W'})_{W'\in \W'_0}, \\
  & \Omega_{\delta_1}, (\encode{W \dleq \alpha \circ W'_1})_{W'_1 \in \W'_1}, \delta_1' \\
  & \vdots \\
  & \Omega_{\delta_m}, (\encode{W \dleq \alpha \circ W'_m})_{W'_m \in \W'_m}, \delta_m' \\
  \end{array}
$$
with $\transform{r} = (r',\Omega_r,\W'_0)$ and $(~ \transform{\delta_j} =
(\delta_j',\Omega_{\delta_j}, \W'_j) ~)_{j=1 \dots m}$.

Then we prove $(f(\ntup{e'}{n},W) \to t)\rho \Leftarrow \Pi$ in $\cflp{\cdom}$ with a {\bf
DF$_\Prog$} root inference step using the program rule $\transform{(R_l)}$ and the substitution
$\theta' = \theta \uplus \rho$ to instantiate the program rule. We next check that every premise of
this inference can be proven in $\CRWLc$:
\begin{itemize}
  \item $\transform{\Prog} \crwl{\cdom} (e'_i\rho \to t_i(\theta \uplus \rho)) \Leftarrow \Pi$ for $i=1 \dots n$. We observe that the only variables of $e'_i$ that can be affected by $\rho$ are those in $\rho_i$. Moreover, $\rho$ cannot affect $t_i$ because the program transformation does not introduce new variables in terms. Therefore $(e'_i\rho \to t_i(\theta \uplus \rho)) = (e'_i \to t_i\theta)\rho_i$ and $\transform{\Prog} \crwl{\cdom} (e'_i \to t_i\theta)\rho_i \Leftarrow \Pi$ for $i=1 \dots n$ follows from inductive hypothesis number 1.
  \item $\transform{\Prog} \crwl{\cdom} (W\rho \to W(\theta \uplus \rho)) \Leftarrow \Pi$. By construction of $\rho$, $(W\!\rho \to W(\theta \uplus \rho)) = d \to d$ and one {\bf RR} inference step proves this statement.
  \item $\transform{\Prog} \crwl{\cdom} (r'(\theta \uplus \rho) \to t\rho) \Leftarrow \Pi$. In this case $t\rho = t$ because $t$ it contains no variables introduced during the transformation, and $r'(\theta\uplus\rho) = r'(\theta\rho'_0)$ since $\rho'_0$ is the only part of $\rho$ that can affect $r'$ and the range of $\theta$ does not include any of the new variables in the domain of $\rho'_0$. Now, $\transform{\Prog} \crwlc (r'\theta \to t)\rho'_0 \Leftarrow \Pi$ follows from inductive hypothesis number 2.
  \item $\transform{\Prog} \crwl{\cdom} \qval{W}(\theta\uplus\rho) \Leftarrow \Pi$. $W$ is a fresh variable and, by construction of $\rho$, $\qval{W}(\theta\uplus\rho) = \qval{d}$. $\transform{\Prog} \crwl{\cdom} \qval{d} \Leftarrow \Pi$ trivially holds.
  \item $\transform{\Prog} \crwl{\cdom} \Omega_r(\theta\uplus\rho) \Leftarrow \Pi$. $\Omega_r(\theta \uplus \rho)$ = $\Omega_r\rho = \Omega_r\rho'_0$ and, by construction, $\rho'_0$ is solution of $\Omega_r$.
  \item $\transform{\Prog} \crwl{\cdom} (\encode{W \dleq \alpha \circ W'})(\theta \uplus \rho) \Leftarrow \Pi$ for each ${W'\in \W'_0}$. We have $(\encode{W \dleq \alpha \circ W'})(\theta \uplus \rho) = (\encode{W \dleq \alpha \circ W'})\rho$ = $\encode{W\rho \dleq \alpha \circ W'\!\rho'_0}$ = $\encode{d \dleq \alpha \circ W'\rho'_0}$. And $\encode{d \dleq \alpha \circ W'\!\rho'_0}$ holds because $d \dleq \alpha \circ d'_0$ by the hypotheses of the inference rule \textbf{QDP$_\Prog$}, and $\encode{d'_0 \dleq W'}$ by inductive hypothesis number 2.
  \item $\transform{\Prog} \crwl{\cdom} \Omega_{\delta_j}(\theta \uplus \rho) \Leftarrow \Pi$  for $j=1 \dots m$. As in the previous premises $\Omega_{\delta_j}(\theta \uplus \rho) = \Omega_{\delta_j}\rho = \Omega_{\delta_j}\rho'_j$ and $\rho'_j$ is solution of $\Omega_{\delta_j}$ as a consequence of the inductive hypothesis number 3.
  \item $\transform{\Prog} \crwl{\cdom} (\encode{W \dleq \alpha \circ W'_j})(\theta \uplus \rho) \Leftarrow \Pi$ for every ${W_j' \in \W'_j}$ and $j=1 \dots m$. We have $(\encode{W \dleq \alpha \circ W'_j})(\theta \uplus \rho) = (\encode{W \dleq \alpha \circ W'_j})\rho = \encode{W\rho \dleq \alpha \circ W'_j\rho} = \encode{d \dleq \alpha \circ W'_j\rho'_j}$. Now, from the hypotheses of the inference rule \textbf{QDP$_\Prog$} follows $d \dleq \alpha \circ d'_j$ for $j=1 \dots m$, and from inductive hypothesis number 3, $\rho'_j$ is solution of $\encode{d'_j \dleq W'_j}$. Hence $\transform{\Prog} \crwl{\cdom} \encode{d \dleq \alpha \circ W_j'\rho_j'} \Leftarrow \Pi$ for $j=1 \dots k$.
  \item $\transform{\Prog} \crwl{\cdom} \delta'_j(\theta \uplus \rho) \Leftarrow \Pi$ for $j=1 \dots m$. In this case $\delta'_j$ can contain variables from both $\theta$ and $\rho'_j$. Hence $\delta'_j(\theta \uplus \rho) = (\delta'_j\theta)\rho'_j$. And $\transform{\Prog} \crwl{\cdom}
(\delta'_j\theta)\rho'_j \Leftarrow \Pi$ follows from the inductive hypothesis number 3.
\end{itemize}

\item {\bf QPF}. In this case $\psi = p(\ntup{e}{n}) \to v$ and the inference step applied at the root is of the form
$$
  \displaystyle\frac
    {(~ (e_i \to t_i)\sharp d_i \Leftarrow \Pi ~)_{i=1 \ldots n}}
    {(p(\ntup{e}{n}) \to v)\sharp d \Leftarrow \Pi}
$$
with $v \in \Var \cup DC^0 \cup B_\cdom$, $\Pi \model{\cdom} p(\ntup{t}{n}) \to v$ and $d \dleq d_i
~ (1 \le i \le n)$. In order to obtain $\transform{(\psi\sharp d \Leftarrow \Pi)}$ one has to:
\begin{itemize}
  \item First, apply the transformation rule {\bf TCE$_1$}, $$\transform{p(\ntup{e}{n})} = (p(\ntup{e'}{n}),\bigcup_{i=1}^n \Omega_i, \bigcup_{i=1}^n \W_i)$$ where $\transform{e_i} = (e_i',\Omega_i,\W_i)$ for $i=1 \dots n$.
  \item Second, apply the transformation rule {\bf TP}, $$\transform{(p(\ntup{e}{n}) \to v)} =  (p(\ntup{e'}{n}) \to v, \bigcup_{i=1}^n \Omega_i, \bigcup_{i=1}^n \W_i) \enspace .$$
  \item And finally, apply the transformation rule {\bf TCS}, $$\transform{(\psi\sharp d \Leftarrow \Pi)} = (p(\ntup{e'}{n}) \to v \Leftarrow \Pi, \bigcup_{i=1}^n \Omega_i \cup \{\encode{d \dleq W} \mid W \in \bigcup_{i=1}^n \W_i\}) \enspace .$$
\end{itemize}

Therefore $$\Omega' = \bigcup_{i=1}^n \Omega_i \cup \{\encode{d \dleq W} \mid W \in \bigcup_{i=1}^n
\W_i\} \enspace .$$

From the premises ${(~ (e_i \to t_i)\sharp d_i \Leftarrow \Pi ~)_{i = 1 \ldots n}}$ of the
inference rule {\bf QPF}, and by the inductive hypothesis we have $\transform{\Prog} \crwl{\cdom}
(e'_i \to t_i)\rho_i \Leftarrow \Pi ~ (1 \le i \le n)$ for some substitutions $\rho_i :
\varset{\Omega'_i} \to \aqdom$ solution of $$\Omega'_i = \Omega_i \cup \{\encode{d_i \dleq W} \mid
W \in \W_i\}$$ for $i=1 \dots n$. We define a new substitution $\rho : \varset{\Omega'} \to \aqdom$
as $\rho = \biguplus_{i=1}^n \rho_i$. It is easy to check that $\rho$ is solution of $\Omega'$:
\begin{itemize}
  \item It is solution of every $\Omega'_i$ for $i=1 \dots n$, since $\rho{\upharpoonright}\varset{\Omega'_i} = \rho_i$. Therefore it is solution of $\bigcup_{i=1}^n \Omega_i$.
  \item It is a solution of $\{\encode{d \dleq W} \mid W \in \bigcup_{i=1}^n \W_i\}$ because as solution of $\Omega'_i$ for $i=1 \dots n$, $\rho$ is solution of $\{\encode{d_i \dleq W} \mid W \in  \W_i\}$, and by the hypothesis of the inference rule {\bf QPF}, $d \dleq d_i ~ (1 \le i \le n)$.
\end{itemize}

We now prove $\transform{\Prog} \crwl{\cdom} (p(\ntup{e'}{n}) \to v)\rho \Leftarrow \Pi$ with a
proof tree $T'$ with a {\bf PF} root inference of the form:
$$
   \displaystyle\frac
    {(~ (e'_i \to t_i)\rho \Leftarrow \Pi ~)_{i = 1\ldots n}}
    {(p(\ntup{e'}{n})\rho \to v) \Leftarrow \Pi}
$$
The rule can be applied because the requirements $v \in \Var \cup DC^0 \cup B_\cdom$ and $\Pi
\model{\cdom} p(\ntup{t}{n}) \to v$ are ensured by the hypothesis of the inference rule {\bf QPF}.
In order to justify that $\transform{\Prog} \crwl{\cdom} (e'_i \to t_i)\rho \Leftarrow \Pi$ for
each $i=1 \dots n$, we observe that the only variables of $(e'_i \to t_i)$ that can be affected by
$\rho$ are those introduced in $e'_i$ by the transformation, and that therefore $(e'_i \to t_i)\rho
= (e'_i \to t_i)\rho_i$ for $i=1 \dots n$, and it is easy to check that these premises correspond
to the inductive hypotheses of this case.

\item {\bf QAC}. This case is analogous to the previous proof, with the only differences being:
\begin{itemize}
\item The inference rule applied at the root of the proof tree is a {\bf QAC} inference rule instead of
       a {\bf QPF} inference rule.
\item In order to obtain the $\transform{(\psi\sharp d \Leftarrow \Pi)}$, the transformation rules
      applied are {\bf TA} and {\bf TCS} instead of {\bf TCE$_1$}, {\bf TP} and {\bf TCS}.
\item The proof tree $T'$ will have an {\bf AC} inference step at its root instead of a {\bf PF}
     inference step.
\end{itemize}
\end{itemize}

\smallskip
\noindent $[\mathit{\ref{theo:transformation:2}}. \Rightarrow
\mathit{\ref{theo:transformation:1}}.]$ \emph{(Transformation soundness).} Assume $\rho \in
\Solc{\Omega'}$ such that $\domset{\rho} = \varset{\Omega'}$ and $\transform{\Prog} \crwl{\cdom}
\psi'\rho \Leftarrow \Pi$ by means of a $\CRWLc$ proof tree $T$ with $k$ nodes. Reasoning by
induction on $k$ we show the existence of a $\QCRWLdc$ proof tree $T'$ witnessing $\Prog \qcrwldc
\psi\sharp d \Leftarrow \Pi$.

\smallskip
\noindent {\em Basis} ($k$=1). If $T$ contains only one node the $\QCRWLdc$ inference step applied
at the root must be any of the following:
\begin{itemize}
\item {\bf TI}. In this case $\psi'\!\rho \Leftarrow \Pi$ is a trivial c-statement. Then $\psi'\!\rho$ is either of the form $e' \to \bot$ or $\Unsat{\cdom}{\Pi}$. In the first case, since the transformation introduces no new variables at the right-hand side of a production, $\psi'$ is of the form $e'' \to \bot$ with $e' = e''\rho$, and $\psi$ is of the form $e \to \bot$, hence $\psi\sharp d \Leftarrow \Pi$ is trivial. Analogously, if $\Unsat{\cdom}{\Pi}$ then $\psi\sharp d \Leftarrow \Pi$ is trivial as well. Therefore $T'$ consists of a single node $\psi\sharp d \Leftarrow \Pi$ with $d$ any value in $\aqdom$, with a {\bf QTI} inference step at its root.

\item {\bf RR}. In this case $\psi'\!\rho = v \to v$ with $v \in \Var \cup B_\cdom$. Then $\psi' = v_1 \to v_2$ for some $v_1, v_2 \in \Var \cup B_\cdom$ such that $\psi'\!\rho = v \to v$. Since $\psi'$ cannot contain new variables introduced by the transformation (by the transformation rules), this means $\psi'\!\rho = \psi'$, and then $\psi' = v \to v$. Therefore $\psi = {v \to v}$, and $T'$ consists of a single node containing $(v \to v)\sharp d \Leftarrow \Pi$ for any $d \in \aqdom$ as the conclusion of a {\bf QRR} inference step.

\item {\bf DC}. Then $\psi'\!\rho  = c \to c$, which means that $\psi'$ can be either of the form $c \to c$, $X \to c$, or $X \to Y$ with $X,Y$ variables. In every case $\psi'$ does not include new variables introduced by the transformation, and therefore $\psi'\rho = \psi'$, which means that $\psi' = c \to c$ is the only possibility. Therefore $\psi = c \to c$, and $T'$ consists of a single node containing $(c \to c)\sharp d \Leftarrow \Pi$ for some $d \in \aqdom$ as the conclusion of a {\bf QDC} inference step.
\end{itemize}

\smallskip
\noindent \emph{Inductive step} ($k{>}$1). The $\CRWLc$ inference step applied at the root must be
any of the following:

\begin{itemize}
\item {\bf DC}. Then $\psi'\!\rho  = c(\ntup{e''}{n}) \to c(\ntup{t}{n})$ where $c \in DC^n$ and $n>0$, which implies that $\psi = c(\ntup{e}{n}) \to c(\ntup{t}{n})$ for values $e_i$ verifying $\transform{e_i} = (e_i', \Omega_i, \W_i)$ for $i = 1 \dots n$, and $e''_i = e'_i \rho$ for $i=1 \dots n$. Then $$\transform{\psi} = \transform{(c(\ntup{e}{n}) \to c(\ntup{t}{n}))} = (c(\ntup{e'}{n}) \to c(\ntup{t}{n}), \union_{i=1}^n \Omega_i, \union_{i=1}^n \W_i)$$ and thus $\varphi = (c(\ntup{e}{n}) \to c(\ntup{t}{n}))\sharp d \Leftarrow \Pi$ for some $d \in \aqdom$ such that $\transform{\varphi} = (\psi' \Leftarrow \Pi, \Omega')$, with $$\Omega' =  \union_{i=1}^n \Omega_i \cup \{\encode{d \dleq W} \mid W \in \union_{i=1}^n \W_i\}$$

The substitution $\rho : \varset{\Omega'} \to \aqdom$ must be solution of $\Omega'$, and the
inference step at the root must be of the form:

$$
  \displaystyle\frac
    {(~ e'_i\rho \to t_i \Leftarrow \Pi ~)_{i = 1\ldots n}}
    {c(\ntup{e'}{n})\rho \to c(\ntup{t}{n}) \Leftarrow \Pi}
$$

In the premises we have the proofs $T_i$ of $\transform{\Prog} \crwl{\cdom} e_i'\rho \Leftarrow
\Pi$ for $i = 1 \dots n$. Now, for each $1 \le i \le n$ we obtain a new value $d_i \in \aqdom$ as
$d_i = \infi \{ W\!\rho \mid W \in \W_i \}$. Then we will prove $\Prog \qcrwl{\qdom}{\cdom}
\varphi$ applying the following {\bf QDC} inference step at the root:
$$
  \displaystyle\frac
    {(~ (e_i \to t_i)\sharp d_i \Leftarrow \Pi ~)_{i = 1 \ldots n}}
    {(c(\ntup{e}{n}) \to c(\ntup{t}{n}))\sharp d \Leftarrow \Pi}
$$

In order to ensure that this step must be applied we must check that $d \dleq d_i ~ (1 \le i \le
n)$. This holds because $\rho$ is solution of $\Omega'$, in particular of $\{\encode{d \dleq W}
\mid W \in \W_i\}$ for $i=1 \dots n$. Therefore for each $i=1 \dots n$ and $W \in \W_i$, $d \dleq
W\!\rho$, which means that $d \dleq d_i = \infi \{ W\rho \mid W \in \W_i\}$. To complete the proof
we must check that there are proof trees for the premises, i.e. that $\Prog \qcrwl{\qdom}{\cdom}
\varphi_i$ with $\varphi_i = (e_i \to t_i)\sharp d_i \Leftarrow \Pi$, $i=1 \dots n$. This is a
consequence of the inductive hypotheses since for each $i=1 \dots n$:
\begin{itemize}
  \item $\transform{\varphi_i} = (e'_i \to t_i \Leftarrow \Pi, \Omega'_i)$, with $\Omega'_i = \Omega_i \cup \{ \encode{d_i \dleq W} \mid W \in \W_i\}$.
  \item $\rho$ is solution of $\Omega'_i$, since it is solution of $\Omega_i$ and by the definition of $d_i$, $d_i \dleq W\!\rho$ for every $W \in \W_i$.
  \item We have that $\transform{\Prog} \crwl{\cdom} e_i'\rho \Leftarrow \Pi$ for $i=1 \dots n$ (the premises of the {\bf DC} step).
\end{itemize}

\item {\bf DF$_\Prog$}. The inference step at the root of $T$ will use an instance $(\transform{R_l})\theta \in [\transform{\Prog}]_\bot$ of a program rule $\transform{R_l}$ of $\transform{\Prog}$. $\transform{R_l}$ will be the transformed of a program rule $R_l = (f(\ntup{t}{n}) \qto{\alpha} r \Leftarrow \ntup{\delta}{m}) \in \Prog$, and therefore will have the form:
$$
  \begin{array}{ll}
    \transform{R_l} = (f(\ntup{t}{n},W) \rightarrow r' \Leftarrow & \qval{W}, \Omega_r, (\encode{W \dleq \alpha \circ W'})_{W'\in \W'_0}, \\
    & \Omega_{\delta_1}, (\encode{W \dleq \alpha \circ W_1'})_{W_1' \in \W'_1}, \delta_1' \\
    & \vdots \\
    & \Omega_{\delta_m}, (\encode{W \dleq \alpha \circ W_m'})_{W_m' \in \W'_m}, \delta_m' \\
  \end{array}
$$
with $\transform{r} = (r',\Omega_r,\W'_0)$ and $(\transform{\delta_j} = (\delta_j',\Omega'_j,
\W'_j))_{j=1 \dots m}$.

In this case, $\psi'\!\rho$ must be of the form $(f(\ntup{e'}{n+1}) \to t)\rho$. By the theorem
premises, there exists a qc-statement $\psi\sharp d \Leftarrow \Pi$ such that
$\transform{(\psi\sharp d \Leftarrow \Pi)} = (\psi' \Leftarrow \Pi,\Omega')$ for some $\Omega'$.
Examining the transformation program rules we observe that the only possibility for $\psi$ is to be
of the form $f(\ntup{e}{n}) \to t$ and that the {\bf TCS} transformation rules should have been
applied followed by {\bf TP} and {\bf TCE$_2$}. This means in particular that $d \neq \bt$ and that
$\transform{e_i} = (e'_i,\Omega_i,\W_i)$  for $i=1 \dots n$ and that $e'_{n+1} = V$ with $V$ fresh
variable. Hence
$$
  \begin{array}{ll}
    \transform{\psi} = (f(\ntup{e'}{n},V) \to t, (& \union_{i=1}^n \Omega_i) \cup \{\qval{V}\} \cup \\
    & \{\encode{V \dleq W'} \mid W' \in \union_{i=1}^n \W_i\}, ~ \{V\}) \\
  \end{array}
$$
and $\varphi = {(f(\ntup{e}{n}) \to t)\sharp d \Leftarrow \Pi}$ for some $d \in \aqdom$. By
hypotheses, $\rho$ is solution of $$ \Omega' = (\union_{i=1}^n \Omega_i) \cup \{\qval{V}\} \cup
\{\encode{V \dleq W'} \mid W' \in \union_{i=1}^n \W_i\} \cup \{\encode{d \dleq V}\}$$ which means,
in particular, that $V\!\rho \in \aqdom$, since it must hold both  $\qval{V}$ and $\encode{d \dleq
V}$.

Therefore the root of $T$ will be $f(\ntup{e'}{n},V)\rho \to t \Leftarrow \Pi$, with premises proof
trees proving:

\begin{enumerate}
  \item $\transform{\Prog} \crwl{\cdom} (~ e'_i\rho \to t_i\theta \Leftarrow \Pi ~)_{i = 1\ldots n}$.
  \item $\transform{\Prog} \crwl{\cdom} (~ V\!\rho \to W\theta \Leftarrow \Pi ~)$. Since $V\!\rho \in \aqdom$ then either $W\theta = V\!\rho$ or $W\theta = \bt$. By premise 4 below, $W\theta \neq \bt$, therefore $W\theta = V\!\rho$.
  \item $\transform{\Prog} \crwl{\cdom} r'\theta \to t \Leftarrow \Pi$.
  \item $\transform{\Prog} \crwl{\cdom} \qval{W\theta} \Leftarrow \Pi$.
  \item $\transform{\Prog} \crwl{\cdom} \Omega_r\theta \Leftarrow \Pi$.
  \item $\transform{\Prog} \crwl{\cdom} (\encode{W \dleq \alpha \circ W'})_{W'\in \W'_0}\theta \Leftarrow \Pi$.
  \item $\transform{\Prog} \crwl{\cdom} \Omega_{\delta_j}\theta \Leftarrow \Pi$ for $j = 1 \dots m$.
  \item $\transform{\Prog} \crwl{\cdom} (\encode{W \dleq \alpha \circ W_j'})_{W_j' \in \W'_j}\theta \Leftarrow \Pi$ for $j = 1 \dots m$.
  \item $\transform{\Prog} \crwl{\cdom} \delta'_j\theta \Leftarrow \Pi$ for $j = 1 \dots m$.
\end{enumerate}

Then we can prove $\Prog \qcrwl{\qdom}{\cdom} \varphi$ by applying a {\bf QDF$_\Prog$} inference
step of the form:
$$
  \displaystyle\frac
    {(~ (e_i \to t_i\theta)\sharp d_i \Leftarrow \Pi ~)_{i = 1 \ldots n}
      \quad (r\theta \to t)\sharp d'_0  \Leftarrow \Pi
      \quad (~ \delta_j\theta\sharp d'_j  \Leftarrow \Pi ~)_{j=1 \ldots m}}
    {(f(\ntup{e}{n}) \to t)\sharp d \Leftarrow \Pi}
$$
where
\begin{itemize}
  \item $d_i = \infi \{W\!\rho \mid W \in \W_i \}$  for $i=1 \dots n$.
  \item $d'_0 = \infi \{W\theta \mid W \in \W'_0 \}$.
  \item $d'_j = \infi \{W\theta \mid W \in \W'_j \}$ for $j=1 \dots m$.
\end{itemize}

For proving $\Prog \qcrwl{\qdom}{\cdom} \varphi$ we need to check that
\begin{itemize}
  \item $d \dleq d_i ~ (1 \le i \le n)$. Since $\rho$ is solution of $\Omega'$, $d \dleq  W\!\rho$, and $W\!\rho \dleq W'\!\rho$ for every $W' \in \W_i$ and every $1 \leq i \leq n$. Therefore $d \dleq \infi \{\rho(W) \mid W \in \W_i \} = d_i$ for $i=1 \dots n$.
  \item $d \dleq \alpha \circ d'_0$. Since $\rho$ is solution of $\Omega'$, $d \dleq V\!\rho = W\theta$. From premise 6, $W\theta \dleq \alpha \circ W'\theta$ for every $W' \in \W'_0$. Therefore $d \dleq \infi \{ W\theta \mid W \in \W'_0 \} = d'_0$.
  \item $d \dleq \alpha \circ d'_j ~ (1 \le j \le m)$. Analogous to the previous point but using premise 8.
\end{itemize}

Finally, in order to justify the premises of the {\bf QDF$_\Prog$} we must prove:
\begin{itemize}
  \item $\Prog \qcrwl{\qdom}{\cdom} (e_i \to t_i\theta) \sharp d_i \Leftarrow \Pi$, which is a consequence of applying the inductive hypotheses to the premises 1, $(~ e'_i\rho \to t_i\theta \Leftarrow \Pi ~)_{i = 1\ldots n}$, following the same reasoning we applied for the premises of the {\bf DC} inference.
  \item $\Prog \qcrwl{\qdom}{\cdom} (r\theta \to t) \sharp d'_0 \Leftarrow \Pi$. Analogously, is a consequence of the inductive hypothesis and of premise 3.
  \item $\Prog \qcrwl{\qdom}{\cdom} (~ \delta_j\theta\sharp d_j' \Leftarrow \Pi ~)_{j=1 \ldots m}$. Again a consequence of the inductive hypothesis, this time applied to the premise 9.
\end{itemize}
\item {\bf PF}. Analogous to the proof for the {\bf DC} inference step.
\item {\bf AC}. analogous to the proof for the {\bf DC} inference step. \qed
\end{itemize}
\end{proof}


Using Theorem \ref{theo:transformation} we can prove that the transformation of goals specified in
Fig. \ref{fig:transformation} preserves solutions in the sense of the following result.

\begin{theorem}
\label{thm:goal-transformation} Let $G$ be a goal for a given $\qcflp{\qdom}{\cdom}$-program
$\Prog$. Then, the two following statements are equivalent:
\begin{enumerate}
  \item $\langle \sigma,\mu,\Pi \rangle \in \Sol{\Prog}{G}$.
  \item $\langle \sigma \uplus \mu \uplus \rho, \Pi \rangle \in \Sol{\transform{\Prog}}{\transform{G}}$ for some $\rho \in \mbox{Val}_\qdom$ such that $\domset{\rho}$ is the set of new variables $W$ introduced by the transformation of $G$.
\end{enumerate}
\end{theorem}
\begin{proof}
Let $G = (~ \delta_i \sharp W_i, W_i \dgeq \beta_i ~)_{i=1\dots m}$, $\sigma$ and $\mu$ be given.
For $i = 1 \dots m$, consider  $\transform{\delta_i} = (\delta'_i, \Omega_i, \W_i)$ and $\Omega'_i
= \Omega_i \cup  \{\encode{W_i \dleq W} \mid W \in \W_i\}$. According to Fig.
\ref{fig:transformation}, $\transform{G} =  (\Omega'_i,\, \qval{W_i},\, \encode{W_i \dgeq
\beta_i},\, \delta'_i)_{i=1\dots m}$. Then, because of Def. \ref{dfn:goals}(2) and the analogous
notion of solution for $\cflp{\cdom}$ goals explained in Sect. \ref{QCFLP}, the two statements of
the theorem can be reformulated as follows:
\begin{enumerate}
  \item[($a$)] $W_i\mu \dgeq \beta_i$ and $\Prog \qcrwldc \delta_i\sigma \sharp W_i\mu \Leftarrow \Pi$ hold for $i = 1 \dots m$.
  \item[($b$)] There exists $\rho \in \mbox{Val}_\qdom$ with   $\domset{\rho} = \bigcup_{i = 1}^m \varset{\Omega_i}$ such that $\rho \in \Sol{\cdom}{\Omega'_i\mu}$, $W_i\mu \dgeq \beta_i$ and $\transform{\Prog} \crwlc (\delta'_i\sigma)\rho \Leftarrow \Pi$ hold for $i = 1 \dots m$.
\end{enumerate}

\noindent [$(a) \Rightarrow (b)$] Assume ($a$). Note that $\transform{\delta_i\sigma \sharp W_i\mu
\Leftarrow \Pi}$ is $\delta'_i\sigma \Leftarrow \Pi, \Omega'_i\mu$. Applying Theorem
\ref{theo:transformation} (with $\psi = \delta_i\sigma$, $d = W_i\mu$ and $\Pi$) we obtain
$\transform{\Prog} \crwlc (\delta'_i\sigma)\rho_i \Leftarrow \Pi$ for some $\rho_i \in
\Sol{\cdom}{\Omega'_i\mu}$ with $\domset{\rho_i} = \varset{\Omega'_i\mu} = \varset{\Omega_i}$. Then
($b$) holds for $\rho = \biguplus_{i=1}^m \rho_i$.

\smallskip
\noindent [$(b) \Rightarrow (a)$] Assume ($b$). Let $\rho_i = \rho {\upharpoonright}
\varset{\Omega_i}$, $i = 1 \dots m$. Note that ($b$) ensures $\transform{\Prog} \crwlc
(\delta'_i\sigma)\rho_i \Leftarrow \Pi$ and $\rho \in \Sol{\cdom}{\Omega'_i\mu}$. Then Theorem
\ref{theo:transformation} can be applied (again with $\psi = \delta_i\sigma$, $d = W_i\mu$ and
$\Pi$) to obtain $\Prog \qcrwldc \delta_i\sigma \sharp W_i\mu \Leftarrow \Pi$. Therefore, ($a$)
holds. \qed
\end{proof}

As an example of goal solving via the transformation, we consider again the \emph{library program}
$\Prog$ and the goal $G$ discussed in the Introduction. Both belong to the instance
$\qcflp{\U}{\rdom}$ of our scheme. Their translation into $\cflp{\rdom}$ can be executed in the
$\toy$ system \cite{toy} after loading the Real Domain Constraints library (\texttt{cflpr}). The
source and translated code are publicly available at \texttt{gpd.sip.ucm.es/cromdia/qlp}. Solving
the transformed goal in $\toy$ computes the answer announced in the Introduction as follows:
\begin{flushleft}
\small
\begin{verbatim}
Toy(R)> qVal([W]), W>=0.65, search("German","Essay",intermediate,W) == R
      { R -> 4 }
      { W=<0.7, W>=0.65 }
sol.1, more solutions (y/n/d/a) [y]? no
\end{verbatim}
\end{flushleft}
The best qualification value for \texttt{W} provided by the answer constraints is  \texttt{0.7}.
\section{Conclusions} \label{Conclusions}

The work in this report is based on the scheme $\cflp{\cdom}$ for functional logic programming with
constraints presented in \cite{LRV07}. Our main results are: a new programming scheme
$\qcflp{\qdom}{\cdom}$ extending the first-order fragment of $\cflp{\cdom}$ with qualified
computation capabilities; a rewriting logic $\QCRWL{\qdom}{\cdom}$  characterizing
$\qcflp{\qdom}{\cdom}$-program semantics; and  a transformation of  $\qcflp{\qdom}{\cdom}$ into
$\cflp{\cdom}$ preserving program semantics and goal solutions, that can be used as a correct
implementation technique. Existing $\cflp{\cdom}$ systems such as $\toy$ \cite{toy} and {\sf Curry}
\cite{curry} that use definitional trees as an efficient implementation tool can easily adopt the
implementation, since the structure of definitional trees is quite obviously preserved by the
transformation.

As argued in the Introduction, our scheme is more expressive than the main related approaches we
are aware of. By means of an example dealing with a simplified library, we have shown that
instances of $\qcflp{\qdom}{\cdom}$ can serve as a declarative language for  flexible information
retrieval problems, where qualified (rather than exact) answers to user's queries can be helpful.

As future work we plan to extend $\qcflp{\qdom}{\cdom}$ and the program transformation in order to
provide explicit support for similarity-based reasoning, as well as the higher-order programming
features available in $\cflp{\cdom}$. We also plan to automate the program transformation, which
should be embedded as part of an enhanced version of the $\toy$ system. Finally, we plan further
research on flexible information retrieval applications, using different instances of our scheme.


\begin{thebibliography}{10}

\bibitem{AEH00}
S.~Antoy, R.~Echahed, and M.~Hanus.
\newblock A needed narrowing strategy.
\newblock {\em Journal of the ACM}, 47(4):776--822, 2000.

\bibitem{Apt90}
K.~R. Apt.
\newblock Logic programming.
\newblock In J.~van Leeuwen, editor, {\em Handbook of Theoretical Computer
  Science}, volume B: Formal Models and Semantics, pages 493--574. Elsevier and
  The MIT Press, 1990.

\bibitem{toy}
P.~Arenas, A.~J. Fern\'andez, A.~Gil, F.~J. L\'opez-Fraguas,
  M.~Rodr\'iguez-Artalejo, and F.~S\'aenz-P\'erez.
\newblock $\mathcal{TOY}$, a multiparadigm declarative language. version 2.3.1,
  2007.
\newblock R. Caballero and J. S\'anchez (Eds.), Available at
  \url{http://toy.sourceforge.net}.

\bibitem{CRR08}
R.~Caballero, M.~Rodr\'iguez-Artalejo, and C.~A. Romero-D\'iaz.
\newblock Similarity-based reasoning in qualified logic programming.
\newblock In {\em PPDP '08: Proceedings of the 10th international ACM SIGPLAN
  conference on Principles and Practice of Declarative Programming}, pages
  185--194, New York, NY, USA, 2008. ACM.

\bibitem{Vad05}
R.~del Vado-V\'irseda.
\newblock Declarative constraint programming with definitional trees.
\newblock In B.~Gramlich, editor, {\em Proceedings of the 5th International
  Conference on Frontiers of Combining Systems (FroCoS'05)}, volume 3717 of
  {\em LNCS}, pages 184--199. Springer Verlag, 2005.

\bibitem{GDL95}
M.~Gabbrielli, G.~M. Dore, and G.~Levi.
\newblock Observable semantics for constraint logic programs.
\newblock {\em Journal of Logic and Computation}, 5(2):133--171, 1995.

\bibitem{GL91}
M.~Gabbrielli and G.~Levi.
\newblock Modeling answer constraints in constraint logic programs.
\newblock In {\em Proceedings of the 8th International Conference on Logic
  Programming (ICLP'91)}, pages 238--252. The MIT Press, 1991.

\bibitem{GMV04}
S.~Guadarrama, S.~Mu{\~{n}}oz, and C.~Vaucheret.
\newblock Fuzzy prolog: A new approach using soft constraint propagation.
\newblock {\em Fuzzy Sets and Systems}, 144(1):127--150, 2004.

\bibitem{curry}
M.~Hanus.
\newblock Curry: an integrated functional logic language, version 0.8.2, 2006.
\newblock M. Hanus (Ed.), Available at
  \url{http://www.informatik.uni-kiel.de/~curry/report.html}.

\bibitem{JMM+98}
J.~Jaffar, M.~Maher, K.~Marriott, and P.~J. Stuckey.
\newblock Semantics of constraints logic programs.
\newblock {\em Journal of Logic Programming}, 37(1-3):1--46, 1998.

\bibitem{Llo87}
J.~W. Lloyd.
\newblock {\em Foundations of Logic Programming, Second Edition}.
\newblock Springer, 1987.

\bibitem{LRV04}
F.~J. L\'opez-Fraguas, M.~Rodr\'iguez-Artalejo, and R.~del Vado-Virseda.
\newblock A lazy narrowing calculus for declarative constraint programming.
\newblock In {\em Proceedings of the 6th International ACM SIGPLAN Conference
  on Principles and Practice of Declarative Programming (PPDP'04)}, pages
  43--54. ACM Press, 2004.

\bibitem{LRV07}
F.~J. L\'opez-Fraguas, M.~Rodr\'iguez-Artalejo, and R.~del Vado-V\'irseda.
\newblock A new generic scheme for functional logic programming with
  constraints.
\newblock {\em Journal of Higher-Order and Symbolic Computation},
  20(1\&2):73--122, 2007.

\bibitem{MP07}
G.~Moreno and V.~Pascual.
\newblock Formal properties of needed narrowing with similarity relations.
\newblock {\em Electronic Notes in Theoretical Computer Science}, 188:21--35,
  2007.

\bibitem{Rie96}
S.~Riezler.
\newblock Quantitative constraint logic programming for weighted grammar
  applications.
\newblock In C.~Retor\'e, editor, {\em Proceedings of the Logical Aspects of
  Computational Linguistics (LACL'96)}, volume 1328 of {\em LNCS}, pages
  346--365. Springer Verlag, 1996.

\bibitem{Rie98phd}
S.~Riezler.
\newblock {\em Probabilistic Constraint Logic Programming}.
\newblock PhD thesis, Neuphilologischen Fakult\"{a}t del Universit\"{a}t
  T\"{u}bingen, 1998.

\bibitem{Rod01}
M.~Rodr\'iguez-Artalejo.
\newblock Functional and constraint logic programming.
\newblock In C.~M. H.~Comon and R.~Treinen, editors, {\em Constraints in
  Computational Logics, Theory and Applications}, volume 2002 of {\em Lecture
  Notes in Computer Science}, pages 202--270. Springer Verlag, 2001.

\bibitem{RR08TR}
M.~Rodr\'iguez-Artalejo and C.~A. Romero-D\'iaz.
\newblock A generic scheme for qualified logic programming.
\newblock Technical Report SIC-1-08, Universidad Complutense, Departamento de
  Sistemas Inform\'aticos y Computaci\'on, Madrid, Spain, 2008.

\bibitem{RR08}
M.~Rodr\'iguez-Artalejo and C.~A. Romero-D\'iaz.
\newblock Quantitative logic programming revisited.
\newblock In J.~Garrigue and M.~Hermenegildo, editors, {\em Functional and
  Logic Programming (FLOPS'08)}, volume 4989 of {\em LNCS}, pages 272--288.
  Springer Verlag, 2008.

\bibitem{Ses02}
M.~I. Sessa.
\newblock Approximate reasoning by similarity-based {SLD} resolution.
\newblock {\em Theoretical Computer Science}, 275(1-2):389--426, 2002.

\bibitem{Sub07}
V.~S. Subrahmanian.
\newblock Uncertainty in logic programming: Some recollections.
\newblock {\em Association for Logic Programming Newsletter}, 20(2), 2007.

\bibitem{Tar55}
A.~Tarski.
\newblock A lattice-theoretical fixpoint theorem and its applications.
\newblock {\em Pacific Journal of Mathematics}, 5(2):285--309, 1955.

\end{thebibliography}

\end{document}